\documentclass[sigconf]{aamas}  

\usepackage{booktabs}
\usepackage{stylesheet}
\usepackage{balance}

\usepackage{flushend}
\setcopyright{ifaamas}  
\acmDOI{doi}  
\acmISBN{}  
\acmConference[AAMAS'20]{Proc.\@ of the 19th International Conference on Autonomous Agents and Multiagent Systems (AAMAS 2020), B.~An, N.~Yorke-Smith, A.~El~Fallah~Seghrouchni, G.~Sukthankar (eds.)}{May 2020}{Auckland, New Zealand}  
\acmYear{2020}  
\copyrightyear{2020}  
\acmPrice{}  

\theoremstyle{acmplain}
\newtheorem{remark}[theorem]{Remark}

\begin{document}

\title{Green Security Game with Community Engagement}  



%
\author{Taoan Huang}
\affiliation{%
  \institution{University of Southern California}
}

\email{taoanhua@usc.edu}

\author{Weiran Shen}
\affiliation{
\institution{Carnegie Mellon University}
}
\email{emersonswr@gmail.com}

\author{David Zeng}
\affiliation{
\institution{Carnegie Mellon University}
}
\email{dzeng@andrew.cmu.edu}

\author{Tianyu Gu}
\affiliation{
\institution{Carnegie Mellon University}
}
\email{tianyug@andrew.cmu.edu
}
\author{Rohit Singh}
\affiliation{
\institution{World Wildlife Fund}
}
\email{rsingh@wwf.sg
}
\author{Fei Fang}
\affiliation{
\institution{Carnegie Mellon University}
}
\email{feif@cs.cmu.edu
}
%
%
%
%
%
%
%

\begin{abstract}  
While game-theoretic models and algorithms have been developed to combat illegal activities, such as poaching and over-fishing, in green security domains, none of the existing work considers the crucial aspect of community engagement: community members are recruited by law enforcement as informants and can provide valuable tips, e.g., the location of ongoing illegal activities, to assist patrols. We fill this gap and (i) introduce a novel two-stage security game model for community engagement, with a bipartite graph representing the informant-attacker social network and a level-$\kappa$ response model for attackers inspired by cognitive hierarchy; (ii) provide complexity results and exact, approximate, and heuristic algorithms for selecting informants and allocating patrollers against level-$\kappa$ ($\kappa<\infty$) attackers; (iii) provide a novel algorithm to find the optimal defender strategy against level-$\infty$ attackers, which converts the problem of optimizing a parameterized fixed-point to a bi-level optimization problem, where the inner level is just a linear program, and the outer level has only a linear number of variables and a single linear constraint. We also evaluate the algorithms through extensive experiments.
\end{abstract}


\maketitle


\section{Introduction}
\label{sec:introduction}
Despite the significance of protecting natural resources to environmental sustainability, a common lack of funding leads to an extremely low density of law enforcement units (referred to as defenders) to combat illegal activities such as wildlife poaching and overfishing (referred to as attacks).
Due to insufficient sanctions, attackers are able to launch frequent attacks \cite{le2006economic,leader1993policies}, making it even more challenging to effectively detect and deter criminal activities through patrolling. To improve patrol efficiency, law enforcement agencies often recruit informants from local communities and plan defensive resources based on tips provided by them \cite{linkie2015editor's}. 
Since attackers are often from the same local community and their activities can be observed by informants through social interactions, such tips contain detailed information about ongoing or upcoming criminal activities and, if known by defenders, can directly be used to guide allocating defensive resources. In fact, community engagement is listed by World Wild Fund for Nature as one of the \textit{six pillars towards zero poaching} \cite{zeropoaching2015b}. The importance of community engagement goes beyond these green security domains about environment conservation and extends to domains such as fighting urban crimes \cite{tublitzfitness,gill2014community}. \ff{add citation}

Previous research in computational game theory have led to models and algorithms that can help the defenders allocate limited resources in the presence of attackers, with applications to enforce traffic \cite{rosenfeld2017security}, combat oil-siphoning \cite{AAAI1816312}, and deceive cyber adversaries \cite{Schlenker:2018:DCA:3237383.3237833} in addition to protecting critical infrastructure \cite{pita2008deployed} and combating wildlife crime \cite{DBLP:journals/aim/FangNPLCASSTL17}. However, none of the work has considered this essential element of community engagement.

Community engagement leads to fundamentally new challenges that do not exist in previous literature. First, the defender not only needs to determine how to patrol but also needs to decide whom to recruit as informants. Second, there can be multiple attackers, and the existence of informants makes the success or failure of their attacks interdependent since any tip about other attackers' actions can change the defender's patrol. Third, because of the combinatorial nature of the tips, representing the defender’s strategy requires exponential space, making the problem of finding optimal defender strategy extremely challenging. Fourth, attackers may notice the patrol pattern over time and adapt their strategies accordingly. 

In this paper, we provide the first study to fill the gap and provide a novel two-stage security game model for community engagement which represents the social network between potential informants and attackers with a bipartite graph. In the first stage of the game, the defender recruits a set of informants under a budget constraint, and in the second stage, the defender chooses a set of targets to protect based on tips from recruited informants. Inspired by the quantal cognitive hierarchy model~\cite{wright2014level}, we use a level-$\kappa$ response model for attackers, taking into account the fact that the attacker can make iterative reasoning and the attacker's strategy will impact the actual marginal strategy of the defender. 


Our second contribution includes complexity results and algorithms for computing optimal defender strategy against level-$\kappa$ ($\kappa<\infty$) attackers. We show that the problem of selecting the optimal set of informants is NP-Hard.
Further, based on sampling techniques, we develop an approximation algorithm to compute the optimal patrol strategy and a heuristic algorithm to find the optimal set of informants to recruit.
For an expository purpose, we mainly describe the algorithms for level-0 attackers and provide the extension to level-$\kappa$ ($0<\kappa<\infty$) attackers in the last section.

The third contribution is a novel algorithm to find the optimal defender strategy against level-$\infty$ attackers, which is an extremely challenging task: an attacker's strategy may affect the defender's marginal strategy, which in turn affects the attackers' strategies and level-$\infty$ attackers is defined through a fixed-point argument; as a result, the defender's utility relies crucially on solving a parameterized fixed-point problem.
A na\"ive mathematical programming-based formulation is prohibitively large to solve. We instead reduce the program to a bi-level optimization problem, where both levels become more tractable. In particular, the inner level optimization is a linear program, and the outer level optimization is one with a linear number of variables and a single linear constraint.

Finally, we conduct extensive experiments. We compare the running time and solution quality of different algorithms. We show that our bi-level optimization algorithm achieves better performance than the algorithm adapted from previous works. We also compare level-0 attackers and the case with insider threat (i.e., the attacker is aware of the informants), where we formulate the problem as a mathematical program and solve it by adapting an algorithm from previous works. We show that the defender suffers from utility loss if the insider threat is not taken into consideration and the defender still assumes a na\"ive attacker model (level-0).

\section{Related Work and Background}
\label{Sec:relatedwork}

Community engagement is studied in criminology. \citet{smith2015poaching,moreto2015introducing,duffy2015militarization} investigate the role of community engagement in wildlife conservation. \citet{linkie2015editor's,gill2014community} show the positive effects of community-oriented strategies. However, they lack a mathematical model for strategic defender-attacker interactions.

Recruitment of informants has also been proposed to study societal attitudes in relation to crimes using evolutionary game theory models. \citet{short2013external} formulate the problem of solving recruitment strategies as an optimal control problem to account for limited resources and budget. In contrast to their work, we emphasize the synergy of community engagement and allocation of defensive resources and aim to find the best strategy of recruiting informants and allocating defensive resources.

In security domains, Stackelberg Security Game (SSG) has been applied to a variety of security problems \cite{tambe2011security}, with variants accounting for alarm systems, surveillance cameras, and drones that can provide information in real time \cite{basilico2017coordinating,Ma_2018,guo2017comparing}. 
Unlike the sensors that provide location-based information as studied in previous works, the kind of tips the informants can provide depends on their social connections, an essential feature about community engagement.

Other than the full rationality model, boundedly rational behavioral models such as quantal response (QR)~\cite{mckelvey1995quantal,yang2012computing} and subjective utility quantal response \cite{nguyen2013analyzing}\ff{add citation here} have been explored in the study of SSG. Our model and solution approach are compatible with most existing behavioral models in the SSG literature, but for an expository purpose, we only focus on the QR model.

\section{Model}
\label{sec:model}
\interfootnotelinepenalty=10000
In this section, we introduce our novel two-stage green security game with community engagement. The key addition is the consideration of informants from local communities. They can be recruited and trained by the defender to provide tips about ongoing or upcoming attacks. 


Following existing works on SSG \cite{jain2010software,korzhyk2011stackelberg}, we consider a game with a set of targets $T=[n]=\{1,\dots,n\}$. The defender has $r$ units of defensive resources and each can protect or cover one target with no scheduling constraint. An attacker can choose a target to attack.
If target $i$ is attacked, the defender (attacker) receives $R^d_i>0$ ($P^a_i<0$) if it is covered, otherwise receives $P^d_i<0$ ($R^a_i>0$).

Informants recruited by the defender can provide tips regarding the exact targets in ongoing or upcoming attacks but tip frequency and usefulness may vary due to heterogeneity in the informants' social connections. We model the interactions and connections between potential informants $X$ (i.e., members of the community that are known to be non-attacker and can be recruited by the defender) and potential attackers $Y$ using a bipartite graph $G_S=(X,Y,E)$ with $X\cap Y=\emptyset$. Here we assume the defender has access to a list of potential attackers which could be provided by the conservation site manager, since the deployment of our work relies on the manager’s domain knowledge, experience, and understanding of the social connections among community members.

When an attacker decides to launch an attack, an informant who interacted with the attacker previously may know his target location.
Formally, for each $v \in Y$, we assume that $v$ will attack a target with probability $p_v$ but the target is unknown without informants and each attacker takes actions independently. An edge $(u,v)\in E$ is associated with an information sharing intensity $w_{uv}$, representing the probability of attack activities of attacker $v$ being reported by $u$, given $v$ attacks and $u$ is recruited as an informant. 


In the first stage, the defender recruits $k$ informants, and in the second stage, the defender receives tips from the informants and allocates $r$ units of defensive resources.
The defender's goal is to maximize the expected utility defined as the summation of the utilities for each attack. 

Let $U$ denote the set of recruited informants in the first stage where $|U|\leq k$, and $V=\{v\mid \exists u\in V, (u,v)\in E\}$ denote the set of attackers that are connected with at least one informant in $U$.
We represent tips as a vector of disjoint subsets of attackers
$\mathbf{V}=(V_1,\ldots, V_n)$, where $V_i$ is the set of attackers who are reported to attack target $i\in T$ such that $V_i\subseteq V, V_i\cap V_j=\emptyset$ for any $i,j\in T$. 
An attacker $v$ is {\it reported} if there exists $i\in T$ such that $v\in V_i$, otherwise he is {\it unreported.} 
We also denote by $V_0=\bigcup_{i\in T}V_i$ the set of reported attackers. 
It is possible that $V_0=\emptyset$ and we say the defender is {\it informed} if $V_0\neq \emptyset$. Note that $\mathbf{V}$ is a compact representation of the tips received by the defender as it neglects the identity of the informants, which is not crucial in the defender's decision making given that all the tips are assumed to be correct.

In practice, tips are infrequent and the defender is often very protective of the informants. Thus, the attackers are often not aware of the existence of informants unless there is a significant insider threat. 
In addition, patrols can be divided into two categories -- routine patrols and ambush patrols, where the latter are in response to tips from informants. Ambush patrols are costly, often requiring rangers to lie in wait for many hours for the possibility of catching a poacher. If not informed, the defender follows her routine patrol strategy $\mathbf{x}_0=(x_1,\ldots,x_n)$ with $x_i$ denoting the probability that target $i$ is covered. Naturally, under this
assumption the defender should use a strategy $x_0$ that is optimal
against the QR model, which can be computed by following \cite{yang2012computing}. If informed she uses different strategies $\mathbf{x}(\mathbf{V})$ based on the tip $\mathbf{V}$. Assume that each attacker, if deciding to attack a target, will respond to the defender's strategy following a known behavioral model -- the QR model.
We define  $\QR(\mathbf{x}'):=(q'_1,\ldots,q'_n)$, where $q'_i$ is the probability of attacking target $i$ defined by
\begin{gather}
\label{eq:quantal_response}
    q'_i=\frac{e^{\lambda \left[x'_iP_i^a+(1-x'_i)R_i^a\right]}}{\sum_{j\in T}e^{\lambda \left[x'_jP_j^a+(1-x'_j)R_j^a\right]}},
\end{gather}
and $\mathbf{x}'$ is the attacker's \emph{subjective belief} of the coverage probabilities.
In the above equation, $\lambda\geq 0$ is the precision parameter \cite{mckelvey1995quantal} fixed throughout the paper. 
We discuss the relaxation of the some of the assumptions mentioned above in Section \ref{sec:conclusion}.

\subsection{Level-$\kappa$ Response Model}
\label{sec:level-k}
Motivated by the costly ambush patrols and inspired by the cognitive hierarchy theory, we propose the level-$\kappa$ response model as the attackers' behavior model.

When the informants' report intensities are negligible, the attackers are almost always faced with the routine patrol $\mathbf{x}_0$. But when the informants' report intensities are not negligible, the attackers' behavior will change the marginal probability that a target is covered. Thus we assume that level-0 attackers just play the quantal response against the routine patrol $\mathbf{x}_0$: $\mathbf{q}^0=\QR(\mathbf{x}_0)$. 
Then the defender will likely get informed with different tips $\mathbf{V}$, and respond with $\mathbf{x}(\mathbf{V})$ accordingly. Over time, the attackers will learn about the change in the frequency that a target is covered. We denote the induced defender's marginal strategy at level 0 by $\hat{\mathbf{x}}^0=\MS(\mathbf{x}_0,\mathbf{x},\mathbf{q}^0)$. 
After observing $\hat{\mathbf{x}}^0$ at level 0, level-1 attackers will update their strategies from $\mathbf{q}^0$ to $\mathbf{q}^1=\QR(\hat{\mathbf{x}}^0)$. 
Similarly, attackers at level $\kappa$ ($0<\kappa< \infty$) will use quantal response against the defender's marginal strategy at level $\kappa-1$, i.e., $\mathbf{q}^{\kappa}=\QR(\mathbf{\hat{x}}^{\kappa-1})$, where $\hat{\mathbf{x}}^{\kappa-1}=\MS(\mathbf{x}_0,\mathbf{x},\mathbf{q}^{\kappa-1})$. 
In Section \ref{levelinf}, we also define level-$\infty$ attackers.


Denote by $\DefEU(U)$ the defender's optimal utility when they recruit a set of informants $U$ and use the optimal defending strategy.
The key questions raised given this model are (i) how to recruit a set $U$ of at most $k$ informants and (ii) how to respond to the provided tips to maximize the expected $\DefEU(U)$?


\section{Defending against Level-0 Attackers}
\label{Sec:level0}

In this section, we first tackle the case where all attackers are level-0 by providing complexity results and algorithms to find the optimal set of informants. Designing efficient algorithms to solve this computationally hard problem is particularly challenging due to the combinatorial nature of the tips and exponentially many possibilities of informant selections. Furthermore, in the general case, attackers are heterogeneous and we do not know which attackers will be reported, making it hard to compute $\DefEU(U)$.
\subsection{Complexity Results}
\newcommand{\MCP}{$\mathsf{MCP}$\xspace}

Let $\mathbf{q}^0=(q_1,\ldots,q_n).$
Before presenting our complexity results, we first define some useful notations. Given the set of informants $U$ and the tips $\mathbf{V}=(V_1,\ldots, V_n)$, we denote by $\tilde{p}_v(V_0)$ the probability of $v\in Y$ attacking a target given $V_0$ such that $V_0=\bigcup_{i\in T}V_i$. We can compute $\tilde{p}_v(V_0)$ with
\begin{gather*}
\tilde{p}_v(V_0)=\begin{cases}    
1 &  v\in V_0\\             
\frac{(1-\tilde{w}_{v})p_v}{(1-\tilde{w}_{v})p_v+1-p_v} & v\in V\setminus V_0\\
p_v & v\in Y\setminus V            
\end{cases},
\end{gather*}
where $\tilde{w}_{v}=1-\prod_{(u,v)\in E,u\in U}(1-w_{uv})$ is the probability of $v$ being reported given he attacks.
Given $V_0$ and $t_i=|V_i|$ reported attacks on each target $i$, we compute the expected utility on $i$ if $i$ is covered with
$\EU_i^c(t_i,V_0):=\left(t_i+q_i\sum_{v\in Y\setminus V_0}\tilde{p}_v(V_0)\right)R_i^d.$
We compute the expected utility if $i$ is uncovered, $\EU_i^u(t_i,V_0)$, similarly.
Then, the expected gain of the target if covered can be written as $\EG_i(t_i,V_0):=\EU_i^c(t_i,V_0)-\EU_i^u(t_i,V_0)$. 

\begin{restatable}[]{theorem}{thmalloc}
\label{thmalloc}
    When the defender is informed by informants $U$, the optimal allocation of defensive resources can be determined in $O(|Y|+n)$ time given the tips $\mathbf{V}=(V_1,\ldots, V_n)$.
\end{restatable}
Given tips from recruited informants, the defender can find the optimal resource allocation by greedily protecting the targets with the highest expected gains. The proof of Theorem \ref{thmalloc} is deferred to Appendix \ref{app:thmalloc}. 
However, the problem of computing the optimal set of informants is still hard.

\begin{restatable}[]{theorem}{thmnphard}
\label{thm:nphard}
	Computing the optimal set of informants to recruit is NP-Hard. 
\end{restatable}
The proof of Theorem \ref{thm:nphard} in Appendix \ref{app:nphard} focuses on a relatively simple case and constructs a reduction from the maximum coverage problem (\MCP).

\subsection{Finding the Optimal Set of Informants}

In this subsection, we develop exact and heuristic informant selection algorithms to compute the optimal set of informants. To find the $U$ that maximizes $\DefEU(U)$, we first focus on computing $\DefEU(U)$ by providing a dynamic programming-based algorithm and approximate algorithms.

\subsubsection{Calculating $\DefEU(U)$}



Let $\DefEU_0$ be the expected utility when using the optimal regular defending strategy against a single attack, which can be obtained by the algorithms introduced in \cite{yang2012computing}. Then
$\DefEU(U)$ can be explicitly written as 
\begin{align*}
&\DefEU(U)=\Pr[V_0=\emptyset]\DefEU_0\\
+&\Pr[V_0\neq\emptyset]\mathsf{E}\left[\sum_{i\in[n]}\bx_i(V)\EG_i(t_i,V_0)+\EU_i(t_i,V_0)\Big|V_0\neq\emptyset\right].    
\end{align*}
To directly compute $\DefEU(U)$ from the above equation is formidable due to the exponential number of tips combinations. However, it is possible to reduce a significant amount of enumeration by handling the calculation carefully. We first develop an Enumeration and Dynamic Programming-based Algorithm (\BF) to compute the exact $\DefEU(U)$ as shown in Algorithm  \ref{AlgoBF}.

First, we compute the utility when the defender is not informed (lines 4-6).
Then, we focus on calculating the total utility $\DefEU'(U)$ in the case when the defender is informed.  
By the linearity of expectation, $\DefEU'(U)$ can be computed as the summation of the expected utility obtained from all targets. Therefore, we focus on the calculation of the expected utility of a single target $i$. 
For each target $i$, Algorithm~\ref{AlgoBF} enumerates all possible types of tips (lines 2-7). We denote each type of tip by a tuple $(t_i,V_0)$, which encodes the set of reported attackers $V_0\neq \emptyset$
and the number of reported attackers $t_i$ targeting location $i$. The probability of receiving $(t_i,V_0)$ can be written as  \[\Pr(t_i,V_0|U)=P_{V_0}{|V_0|\choose t_i}q_i^{t_i}(1-q_i)^{|V_0|-t_i},\] where 
\begin{eqnarray}\label{eqnpv0}
    P_{V_0}=\prod_{v\in V_0}(\tilde{w}_{v}p_v) \prod_{v\in V\setminus V_0}(1-\tilde{w}_{v}p_v)
\end{eqnarray}
is the probability of having $V_0$ being the set of reported attackers given $U$ (line 3).
Let $P_{i,r}$ be the probability of $i$ being among the $r$ targets with the highest expected gain given $(t_i, V_0)$ and $U$ (lines 12-13). For a given tip type $(t_i, V_0)$, the expected contribution to $\DefEU'(U)$ of target $i$ is
\begin{eqnarray*}
&\Pr(t_i,V_0|U)\cdot \EU_i(t_i,V_0)+P_{V_0}{|V_0|\choose t_i}q_i^{t_i}\cdot P_{i,r}\EG_i(t_i,V_0)\\
&=P_{V_0}{|V_0|\choose t_i}q_i^{t_i}\left((1-q_i)^{|V_0|-t_i}\EU_i(t_i,V_0)+P_{i,r} \EG_i(t_i,V_0)\right).
\end{eqnarray*}
We can then compute $\DefEU'(U)$ by summing over all possible $t_i, V_0 \neq \emptyset$.



The calculation of $P_{i,r}$ is all that remains. 
This can be done very efficiently via Algorithm \ref{AlgoDP}, a dynamic programming-based calculation. 
Let $\{{i_1},\ldots,{i_{n-1}}\}$ denote the set of targets apart from $i$, i.e., $T\setminus \{i\}$ (line 1) and $y!\cdot f(s,x,y)$
be the probability of having $y$ reported attacks among the first $s$ targets with $x$ of the targets having expected gain higher than $\EG_i$ given the tips of type $(t_i,V_0)$. Therefore, $f(s,x,y)$ can be neatly written as 
\begin{eqnarray*}
	f(s,x,y)&=&\sum_{\substack{a_1+\cdots+a_{s}=y,\\\sum_{j=1}^{s}\mathbf{1}_{[\EG_{i_j}(a_j,V_0)>\EG_i(t_i,V_0)]}=x}}\frac{q_{i_1}^{a_1}q_{i_2}^{a_2}\cdots q_{i_{s}}^{a_{s}}}{a_1!a_2!\cdots a_{s}!},
\end{eqnarray*}
which can be calculated using dynamic programming (line 5-11). 
Computing $f(s,x,y)$ is done in a similar way by counting the number of $s$-partitions on integer $y$, where we also consider the constraint brought in by the limitation on the number of resources. To calculate $f(s,x,y)$, we enumerate $a_s$ as $\tilde{y}$ (line 6) and compare $\EG_{i_s}(a_s,V_0)$ with $\EG_i(t_i,V_0)$ (line 8). If $\EG_{i_s}(a_s,V_0)>\EG_i(t_i,V_0)$, we check the value of $f(s-1,x-1,y-\tilde{y})$ (line 9), otherwise check $f(s-1,x,y-\tilde{y})$ (line 11).
Thus, we have $P_{i,r}=(|V_0|-t_i)! \left(\sum_{x=0}^{r-1}f(s,x,|V_0|-t_i)\right).$
The time complexity for Algorithm \ref{AlgoDP} is $O(nr|Y|^2)$ and $O(2^{|Y|}n^2r|Y|^3)$ for Algorithm \ref{AlgoBF}.


\begin{algorithm}[t]
	\caption{Calculate $\DefEU(U)$}\label{AlgoBF}
	\begin{algorithmic}[1]
		\State $\EU\gets 0$
		\For{all possible sets of reported attackers $V_0\subseteq V$}
		\State $P_{V_0}\gets \prod_{v\in V_0}(\tilde{w}_{v}p_v) \prod_{v\in V\setminus V_0}(1-\tilde{w}_{v}p_v)$
		\If {$V_0=\emptyset$}
		\State $\EU=\EU+P_{V_0}\sum_{v\in Y}\tilde{p}_v(V_0)\DefEU_0$
		\State Continue to line 2
		\EndIf
		\For {target $i\in T$ and $0\leq t_i\leq |V_0|$}
		
		\State Calculate $f(\cdot)$ given $|V_0|,i,t_i$
		\State $\EG_i\gets(t_i+q_i\sum_{v\in Y\setminus V_0}\tilde{p}_v(V_0)) (R_i^d-P_i^d)$
		\State $\EU_i^u\gets (t_i+q_i\sum_{v\in Y\setminus V_0}\tilde{p}_v(V_0))P_i^d$
		\State $P_{i,r}\gets (|V_0|-t_i)! \left(\sum_{x=0}^{r-1}f(s,x,|V_0|-t_i)\right)$
		\State $\EU=\EU+P_{V_0} {|V_0| \choose t_i}q_i^{t_i}\cdot P_{i,r}\cdot \EG_i$
		\State $\EU=\EU+P_{V_0} {|V_0| \choose t_i}q_i^{t_i}(1-q_i^{t_i})\EU_i^u$
		
		\EndFor
		\EndFor
		\State $\DefEU(U)\gets \EU$
	\end{algorithmic}
\end{algorithm}

\begin{algorithm}[t]
    \caption{Calculate $f(\cdot)$ given $|V_0|,i,t_i$}\label{AlgoDP}
    \begin{algorithmic}[1]
    \State $\{{i_1},\ldots,{i_{n-1}}\}\gets T\setminus \{i\}$
    \State $\EG_i\gets(t_i+q_i\sum_{v\in Y\setminus V_0}\tilde{p}_v(V_0)) (R_i^d-P_i^d)$
    \State Initialize $f(s,x,y)\gets 0$ for all $s,x,y$ 
		\State $f(0,0,0)\gets 1$
		\For{$s$ in $[1,n-1]$, $x$ in $[0, \min(s,r)]$, $y$ in $[0, |V_0|-t_i]$}
		\For{$\tilde{y}$ in $[0, y]$}
		\State $\EG_{i_s}\gets (\tilde{y}+q_{i_s}\sum_{v\in Y\setminus V_0}\tilde{p}_v(V_0)) (R_{i_s}^d-P_{i_s}^d)$
		\If {$\EG_{i_s}>\EG_i$}
		\State $f(s,x,y) \pluseq \frac{q_{i_s}^{\tilde{y}}}{\tilde{y}!} f(s-1,x-1,y-\tilde{y})$
		\Else 
		\State $f(s,x,y) \pluseq \frac{q_{i_s}^{\tilde{y}}}{\tilde{y}!}f(s-1,x,y-\tilde{y})$
		\EndIf
		\EndFor  
		\EndFor
    \end{algorithmic}
\end{algorithm}

Since \BF runs in exponential time, we introduce approximation methods to estimate $\DefEU(U)$. 
Let $\DefEU(U,C)$ be the estimated defender's utility returned by Algorithm \ref{AlgoBF} if only subsets of reported attackers $V_0$ with $|V_0|< C$ are enumerated in line 2. We denote by \Ctrunc this approach of estimating $\DefEU(U)$.
Next, we show that $\DefEU(U,C)$ is close to the exact $\DefEU(U)$ when it is unlikely to have many attacks happening at the same time. Formally, assume that the expected number of attacks is bounded by a constant $C'$, that is $\sum_{v\in Y}p_v\leq C'$, $\DefEU(U,C)$ for $C>C'$ is an estimation of $\DefEU(U)$ with bounded error.

\begin{restatable}[]{lemmma}{lemmabound}
\label{lemmabound}
	Assume that $\sum_{v\in Y}p_v\leq C'$ and $|P_i^d|, |R_i^d|\leq Q$, the error of estimation $|{\DefEU}(U,C)-\DefEU(U)|$ for $C> C'$ is at most:
	\begin{gather*}
	    Q \cdot e^{-2(C-C')^2/|Y|}\left( C+\frac{1}{1-e^{-4(C-C')/|Y|}}\right).
	\end{gather*}
\end{restatable}
The proof of Lemma \ref{lemmabound} is deferred to Appendix \ref{app:ctruncated-proof}.
The time complexity of \Ctrunc is given by $O(n^2r|Y|^{C+3})$.

However, for the case where $\sum_{v\in Y}p_v$ is large, we have to set $C$ to be larger than $\sum_{v\in Y}p_v$ for \Ctrunc in order to obtain a high-quality solution; otherwise the error will become unbounded. 
To mitigate this limitation, we also propose an alternative sampling approach, \Sampling, to estimate $\DefEU(U)$ for general cases without restrictions on $\sum_{p_v}$. Instead of enumerating all possible $V_0$ as \BF does, in \Sampling, we draw $\mathsf{T}$ i.i.d. samples of the set of reported attackers where each sample $V_0$ is drawn with probability $P_{V_0}$. \Sampling takes the average of the expected defender's utility when having $V_0$ as the reported attackers
over all samples as the estimation of $\DefEU(U)$. We can sample $V_0$ as follows: (i) Let $V_0=\emptyset$ initially;
    (ii) For each $v\in V$, add $v$ to $V_0$ with probability $\tilde{w}_vp_v$;
    (iii) Return $V_0$ as a sample of the set of reported attackers.
From Equation (\ref{eqnpv0}), the above sampling process is consistent with the distribution of $V_0$. \Sampling returns an estimation of $\DefEU(U)$ in $O(\mathsf{T}n^2r|Y|^3)$ time. 

\begin{proposition}
    Let $\DefEU^{(\mathsf{T})}(U)$ be the estimation of $\DefEU(U)$ given by \Sampling using $\mathsf{T}$ samples. We have:
    $\lim_{\mathsf{T} \rightarrow \infty}\DefEU^{(\mathsf{T})}(U)=\DefEU(U)$
\end{proposition}

\subsubsection{Selecting Informants $U$}



Given the algorithms for computing $\DefEU(U)$, a straightforward way of selecting informants is through enumeration (denoted as \Select).

When using \Ctrunc as a subroutine to compute $\DefEU(U)$, the solution quality of the selected set of informants is guaranteed by the following theorem. 

\begin{theorem}
	Assume that $\sum_{v\in Y}p_v\leq C'$ and $|P_i^d|, |R_i^d|\leq Q$. Let $U_{\OPT}$ and $U'$ be the optimal set of informants and the one chosen by  \Ctrunc. Then for $C>C'$, the error $|\DefEU(U_{\OPT})-\DefEU(U')|$ can be bounded by:
    \begin{gather*}
   2Q \cdot e^{-2(C-C')^2/|Y|}\left( C+\frac{1}{1-e^{-4(C-C')/|Y|}}\right).    
    \end{gather*}
\end{theorem}




\begin{proposition}
    Using \Sampling to estimate $\DefEU$, the optimal set of informants can be found when $\mathsf{T}\rightarrow \infty$.
\end{proposition}

\begin{algorithm}[!htbp]
	\caption{$\mathsf{Search}(U')$}\label{AlgoHeuristic}
	\begin{algorithmic}[1]
	    \If {$|U'|=k$}
	        \State Update $\OPT$ with $(U', \DefEU(U'))$
	        \State \Return
	    \EndIf
	    \State $u_1\gets \arg\max_{u\in X}\DefEU(U'\cup \{u\})$
	    \State $u_2\gets \arg\max_{u\in X\setminus \{u_1\}}\DefEU(U'\cup \{u\})$
        \State $\mathsf{Search}(U'\cup \{u_1\})$, $\mathsf{Search}(U'\cup \{u_2\})$
	\end{algorithmic}
\end{algorithm}

Based on existing results in submodular optimization ~\cite{nemhauser1978analysis}, one may expect a greedy algorithm that step by step adds an informant that leads to the largest utility to work well. However, the set function $\DefEU(U)$ in our problem violates submodularity (see Appendix \ref{app:submodular}) and such greedy algorithm will not guarantee an approximation ratio of $1-1/e$. 
Therefore, we propose \Heuristic (Greedy-based Search Algorithm) for the selection of informants as shown in Algorithm \ref{AlgoHeuristic}. \Heuristic starts by calling $\mathsf{Search}(\emptyset)$. While $|U'|<k$, $\mathsf{Search}(U')$ expands the current set of informants $U'$ by adding $u_1,u_2$ to $U'$ and recursing, where $u_1$ and $u_2$ are the two informants that give the largest marginal gain in $\DefEU$ (line 4-5); Otherwise, it updates the optimal solution with $U'$ (line 1-3).


We identify a tractable case to conclude the section.
\begin{restatable}[]{lemmma}{lemSISI}
\label{lem:SISI}
	Given the set of recruited informants $U$, the defender's expected utility $\DefEU(U)$ can be computed in polynomial time if $w_{uv}=1\, \forall (u,v)\in E$. When $k$ is a constant, the optimal set of informants can be computed in polynomial time.
\end{restatable}

This represents the case where the informants have strong connections with a particular group of attackers and can get full access to their attack plans. We refer to the property of $w_{uv}=1$ for all $u,v$ as SISI (Strong Information Sharing Intensity).
Denote by \SpeCase (Algorithm for SISI) the polynomial-time algorithm in Lemma \ref{lem:SISI}. We provide more details about the SISI case in Appendix \ref{app:prooflemSISI} and \ref{app:asisi}.

We summarize the time complexity of all algorithms for computing the optimal $U$ in Table~\ref{table2} in the Appendix.

\section{Defending Against Level-$\infty$ Attackers}
\label{levelinf}

As discussed in Section \ref{sec:level-k}, 
a level-$\kappa$ attacker may keep adapting to the new marginal strategy formed by his current level of behavior. In this section, we first show in Theorem \ref{thm:fixedpoint} that there exists a fixed-point strategy for the attacker in our level-$\kappa$ response model, and then use that to define the level-$\infty$ attackers.  

We formulate the problem of finding the optimal defender's strategy for this case as a mathematical program. However, such a program can be too large to solve. We propose a novel technique that reduces the program to a bi-level optimization problem, with both levels much more tractable.



\begin{restatable}[]{theorem}{thmfixedpoint}
\label{thm:fixedpoint}
	Let $\Delta_n=\{\mathbf{q}\mid\mathbf{q}\in[0,1]^{n},\mathbf{1}^{\mathsf{T}}\mathbf{q}\leq 1\}$. Given defender's strategies $\mathbf{x}_0$ and $\mathbf{x}(\mathbf{V})$, there exists $\mathbf{q}^*\in \Delta_n$ such that $\mathbf{q}^*=\QR(\MS(\mathbf{x}_0,\mathbf{x},\mathbf{q}^*))$.
\end{restatable}
\begin{proof}
Since $\Delta_n$ is a compact convex set and $\QR(\MS(\mathbf{x}_0,\mathbf{x},\mathbf{q}^*))$ is a continuous function of $\mathbf{q}$, by Brouwer fixed-point theorem, there exists $\mathbf{q}^*\in \Delta_n$ such that $\mathbf{q}^*=\QR(\MS(\mathbf{x}_0,\mathbf{x},\mathbf{q}^*))$.
\end{proof}
According to the definition of level-$\kappa$ attackers, we have $\mathbf{q}^{\kappa+1}=\QR(\MS(\mathbf{x}_0,\mathbf{x},\mathbf{q}^{\kappa}))$. Slightly generalizing the definition, we define a level-$\infty$ attacker as:
\begin{definition}[level-$\infty$ attacker]
\label{def:level_inf}
    Given the defender's strategies $\mathbf{x}_0$ and $\mathbf{x}(\mathbf{V})$, the strategy $\mathbf{q}$ of a level-$\infty$  attacker satisfies $\mathbf{q}=\QR(\MS(\mathbf{x}_0,\mathbf{x},\mathbf{q}))$.
\end{definition}


\begin{remark}
Note that Definition \ref{def:level_inf} is not obtained by taking the limit of the level-$\kappa$ definition, since such a limit may not even exist (see Example \ref{example:non_convergence} in Appendix \ref{app:counter-example}). 
\end{remark}

\begin{remark}

Although the level-$\infty$ attacker is defined through a fixed point argument, we still stick to the Stackelberg assumption: the defender leads and the attacker follows. Notice that in the equation $\mathbf{q}=\QR(\MS(\mathbf{x}_0,\mathbf{x},\mathbf{q}))$, $\mathbf{q}$ will only be defined after the defender commits to strategies $\mathbf{x}_0$ and $\mathbf{x}$. However, it is different from the standard Strong Stackelberg Equilibrium \cite{korzhyk2011stackelberg} in that the attacker is following a level-$\infty$ response model, as defined by the fixed point equation. 

Also, as we will discuss in Section \ref{sec:bi_level_comparison} on our experiments, when $r=n$, the defender's optimal strategy is not to use up all the available resources. This is clearly different from a Nash equilibrium, as the defender still has incentives to use more resources.
\end{remark}

\subsection{Convergence Condition for the Level-$\kappa$ Response Model}

We focus on the single-attacker case, where there are only $n$ different types of tips. We use $\mathbf{V}_i$ to denote the tips where the attacker is reported to attack target $i$. When the attacker is using strategy $q$, the probability of receiving $\mathbf{V}_i$ is $\Pr\{\mathbf{V}_i\}=wq_i$.

\begin{theorem}
	\label{thm:convergence}
	Let $\bar{x}_i=\max_j\{x_i(\mathbf{V}_j)\}$. In the single attacker case, if there exists constant $L\in [0,1)$, such that $\bar{x}_i\le \frac{L}{n\lambda(R^a_i-P^a_i)},\forall i$, then level-$\kappa$ agents converge to level-$\infty$ agents as $\kappa$ approaches infinity.
\end{theorem}

The proof of Theorem \ref{thm:convergence} is omitted since it is immediate from the following lemma:
\begin{restatable}[]{lemmma}{lemLipschitz}
	\label{lem:lipschitz}
	In the single attacker case, if there exists constant $L\in [0,1)$, such that $\bar{x}_i\le \frac{L}{n\lambda(R^a_i-P^a_i)}$ for all $i$, then $g(\mathbf{q})$ is $L$-Lipschitz with respect to the $L^1$-norm, i.e., $g(\mathbf{q})$ is a contraction.
\end{restatable}
The proof of Lemma \ref{lem:lipschitz} is deferred to Appendix \ref{app:ipschitz-proof}.

\begin{corollary}
In the single attacker case, if there exists a constant $L\in[0,1)$, such that $\frac{L}{n\lambda(R^a_i-P^a_i)}>1, \forall i$, then level-$\kappa$ agents converge to level-$\infty$ agents as $\kappa$ goes to infinity.
\end{corollary}

\subsection{A Bi-Level Optimization for Solving the Optimal Defender's Strategy }
In this section, we still consider the single attacker case and assume the defender has $r \geq 1$ resources. Clearly, the optimal set of informants should contain the ones with the highest information sharing intensities. It remains to compute the optimal strategies $\mathbf{x}_0$ and $\mathbf{x}(\mathbf{V})$.
Given the optimal set of informants $U^*$, the probability of receiving a tip is $w = 1 - \prod_{u \in U^*} (1 - w_{u1})$.
Let $\Pr\{\mathbf{V}\}$ be the probability of receiving tips $\mathbf{V}$, which depends $\mathbf{q}$. Let $\mathbf{x}(\mathbf{V})=(x_1(\mathbf{V}),\ldots,x_n(\mathbf{V}))$ be the defender strategy when receiving tips $\mathbf{V}$.

Let $\mathbf{q}=(q_1,\ldots,q_n)$ be the strategy of the level-$\infty$ attacker.
Given $\mathbf{V}$ and the corresponding $t_i$'s, the expected number of attackers that are going to attack target $i$ is $d_i=t_i+(1-\sum_j t_j)\tilde{p}_v(\emptyset)q_i$.
Therefore, given $\hat{\bx}$ we have the defender's expected utility $\DefEU(\mathbf{x}_0,\mathbf{x})$ as
\begin{gather*}
     \DefEU(\mathbf{x}_0,\mathbf{x})=\textstyle\sum_{\mathbf{V},i}\Pr\{\mathbf{V}\}d_i\left[ P_i^d+x_i(\mathbf{V})\left( R_i^d-P_i^d \right) \right].
\end{gather*}
Then the problem of finding the optimal defender strategy can be formulated as the following mathematical program:
\begin{align*}
    \text{max\quad} \DefEU(\mathbf{x}_0,\mathbf{x})
\quad \text{s.t.\quad} \mathbf{q}=\QR(\MS(\mathbf{x}_0,\mathbf{x},\mathbf{q})).
\end{align*}
In the single-attacker case, we need $n$ and $n^2$ variables to represent $\mathbf{x}_0$ and $\mathbf{x}$. We can use the \texttt{QRI-MILP} algorithm\footnote{An algorithm that computes an approximate defender's optimal strategy against a variant of level-0 attackers who take into account the impact of informants when determining the target they attack. See Appendix \ref{app:infawareatt} for more details.} to find the solution. However, this approach needs to solve a mixed integer program and does not scale well. 


To tackle the problem, we focus on the defender's marginal strategy instead of the full strategy representation, and decompose the above program into a bi-level optimization problem. 

Let $\hat{\mathbf{x}}=\MS(\mathbf{x}_0,\mathbf{x},\mathbf{q})=\sum_{\mathbf{V}}\Pr\{\mathbf{V}\}\mathbf{x}(\mathbf{V})$, where we slightly abuse notation and use $\mathbf{V}=\emptyset$ to denote the case of receiving no tip, $\mathbf{x}(\emptyset)$ to denote $\mathbf{x}_0$. 
The bi-level optimization method works as follows. At the inner level, we fix an arbitrary feasible $\hat{\mathbf{x}}$, and solve the following mathematical program:
\begin{align*}
    \text{max\quad}& \DefEU(\hat{\mathbf{x}})\nonumber\\
    \text{s.t.\quad}&\textstyle\sum_{\mathbf{V}}\Pr\{\mathbf{V}\}\bx(\mathbf{V})=\hat{\mathbf{x}},\,\mathbf{q}=\QR(\hat{\bx})\\
    &\mathbf{1}^{\mathsf{T}}\bx(\mathbf{V})\le r, \bx(\mathbf{V})\in [0,1]^n, \forall \mathbf{V}
\end{align*}
Since $\hat{\mathbf{x}}$ is fixed, $\mathbf{q}$ and $\Pr\{\mathbf{V}\}$ are also fixed. Thus, the  program above becomes a linear program, with $\mathbf{x}(\mathbf{V})$ as variables. We can always find a feasible solution to it by simply setting $\mathbf{x}(\mathbf{V})=\hat{\mathbf{x}}, \forall \mathbf{V}$. 
Solving this linear program gives us the optimal defender's utility $\DefEU(\hat{\mathbf{x}})$ for any possible $\hat{\mathbf{x}}$. To find the optimal defender strategy, we solve the outer-level optimization problem below:
\begin{align*}
    \text{max\quad} \DefEU(\hat{\mathbf{x}})
    \text{\quad s.t.\quad} \hat{\mathbf{x}} \text{ is feasible.}
\end{align*}
Since the feasible region of $\hat{\mathbf{x}}$ is continuous, we can use any known algorithm (e.g., gradient descent) to solve the outer-level program.
The inner-level linear program still suffers from the scalability problem. However, when there are multiple attackers, the optimal objective value can be well-approximated by simply sampling a subset of possible $\mathbf{V}$'s, or focusing only on the $\mathbf{V}$'s with the highest probabilities. For those $\mathbf{V}$'s that are not considered, we can always use $\mathbf{x}_0$ as the default strategy for $\mathbf{x}(\mathbf{V})$.
\section{Defending Against Informant-Aware Attackers}
\label{sec:informant-aware-attackers}
We now consider a variant of our model where attackers take into account the impact of informants when determining the target they attack. Specifically, we assume the attackers follow the QR behavior model but incorporate the probability of being discovered when determining their expected utility for attacking a target.\footnote{Consider attackers that have had experience playing against the defender. Over time, the attacker might start to consider their expected utility in practice, which is affected by informants.} In this setting, the attackers' subjective belief $\mathbf{x}'$ of the target coverage probability does not necessarily satisfy $\sum_i x'_i\le r$. Consider the example of a single attacker and a single informant with report intensity 1. Assume that the defender has $r=1$ and always protects the target being reported with probability 1. Then no matter which target the attacker chooses to attack, it will always be covered.

We focus on the single attacker case with $r \geq 1$. We first consider the problem of computing the optimal defender strategy when given the set of informants $U$ and associated probability of receiving a tip $w = 1 - \prod_{u \in U} (1 - w_{u1})$. In the general case with multiple attackers, we will need to specify the defender strategy for each combination of tips received. However, when there is only one attacker, we can succinctly describe the defender strategy by their default strategy without tips, $\mathbf{x}$, and their probability of defending a location after receiving a tip for that location, $\mathbf{z}$. Then, under the QR adversary model, the probability $q_i$ of the attacker targeting location $i$ will be:
\begin{gather*}
    q_i = \frac{e^{\lambda\left\{\left[(1 - w)x_i + wz_i\right]P^a_i + \left[1 - (1 - w)x_i - wz_i\right]R^a_i\right\}}}{\sum_{j \in \calT}e^{\lambda\left\{\left[(1 - w)x_j + wz_j\right]P^a_j + \left[1 - (1 - w)x_j - wz_j\right]R^a_j\right\}}}.
\end{gather*}

This leads to the following optimization problem, \texttt{QRI}, to compute the optimal defender strategy:

\begin{align}
    \max_{x, y, a} & \quad \frac{\sum_{i \in \calT} e^{\lambda R^a_i}e^{-\lambda(R^a_i - P^a_i)y_i}\left[(R^d_i - P^d_i)y_i + P^d_i\right]}{\sum_{i \in \calT} e^{\lambda R^a_i}e^{-\lambda(R^a_i - P^A_i)y_i}} \nonumber\\
    \text{subject to} & \quad y_i = (1 - w)x_i + wz_i, \quad  \forall i \in \calT \label{optln:y}\\
    & \quad \textstyle\sum_{i \in \calT} x_i \leq r \\
    & \quad 0 \leq x_i,z_i \leq 1, \quad \forall i \in \calT \label{optln:z}
\end{align}

We can compute the optimal defender strategy by adapting the approach used in the \Pasaq algorithm~\cite{yang2012computing}. The description of the algorithm is deferred to Appendix \ref{app:infawareatt}.

\section{Experiment}
\label{Sec:experiment}

In this section, we demonstrate the effectiveness of our proposed algorithms through extensive experiments. 
In our experiments, all reported results are averaged over 30 randomly generated game instances. See Appendix \ref{app:expsetup} for details about generating game instances and parameters. Unless specified otherwise, all game instances are generated in this way.

\subsection{Experimental Results}
We compare the scalability and the solution quality of \Select using \BF, \Ctrunc, \Sampling to obtain $\DefEU$  and \Heuristic for different settings of the problems against level-0 attackers.

First, we test the case where $\sum_{v\in Y}p_v<3$.
We set $|X|=6, k=4, n=8, r=3$ and enumerate $|Y|$ from 2 to 16.
The results are shown in Figure \ref{figbound}. We also include \Greedy as a baseline that always chooses the  informants that maximizes the probability of receiving tips. We can see that \Sampling performs the best in term of runtime, but fails to provide high-quality solutions. While \Ctrunc is slower than \Sampling, it performs the best with no error on all test cases.
However, when there is no restriction on $\sum_{v\in Y}p_v$, as shown in  Figure \ref{figunbound}, \Ctrunc performs badly, even worse than \Greedy for large $|Y|$,
while \Sampling performs a lot better and \Heuristic performs the best.
We also fix $|X|=7$, $|Y|=10, k=3, r=5$ and change the number of targets $n$ from 5 to 25 for $\sum_{v\in Y}p_v<3$. The results are shown in Figure \ref{figtargets}. \Heuristic is the fastest but provides slightly worse solutions than \Ctrunc does. The runtime of \Greedy is less than 0.3s for all instances tested.

We then perform a case study to show the trade-off between the optimal number of resources to allocate and the optimal number of informants to recruit with budget constraints when defending against level-0 attackers. We set $|X|=|Y|=n=6$ and generate an instance of the game. We set the cost of allocating one defensive resource $C_r = 3$ and the cost of hiring one informant $C_i=1$. Given a budget $B$, we can recruit $k$ informants and allocate $r$ resources when $k\cdot C_i+r\cdot C_r\leq B$. The trade-off between the optimal $k$ and $r$ is shown in Figure \ref{figbudcons}.
In the same instance, we study how the defender's utility would change by  increasing the number of recruited informants with fixed $r$. Given a fixed number of resources, the defender should recruit as many informants as possible. We can also see that assuming the defender can acquire sufficient resources, the importance of recruiting additional informants is diminished.
This result provides useful guidance to defenders such as conservation agencies in allocating their budget and recruiting informants.

We run additional experiments for the SISI case and do a case study to show the errors of the estimations for all $U\subseteq X$ on 2 instances. We present the results in Appendix \ref{app:additionalexp}.

\begin{figure*}[hbtp]
	\centering
	\begin{subfigure}[htbp]{0.48\textwidth}
		\centering
		\includegraphics[width=3.7cm]{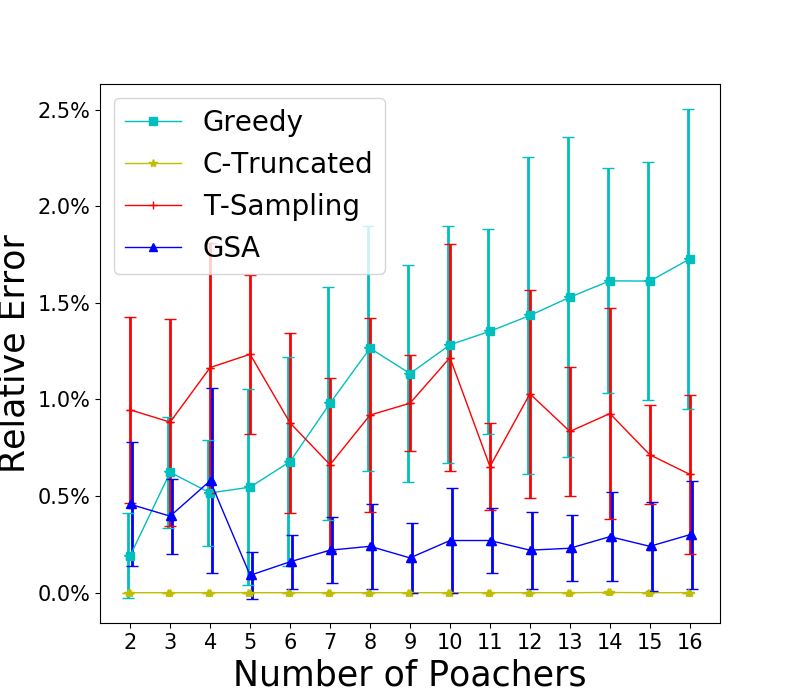}
		~
		\includegraphics[width=3.7cm]{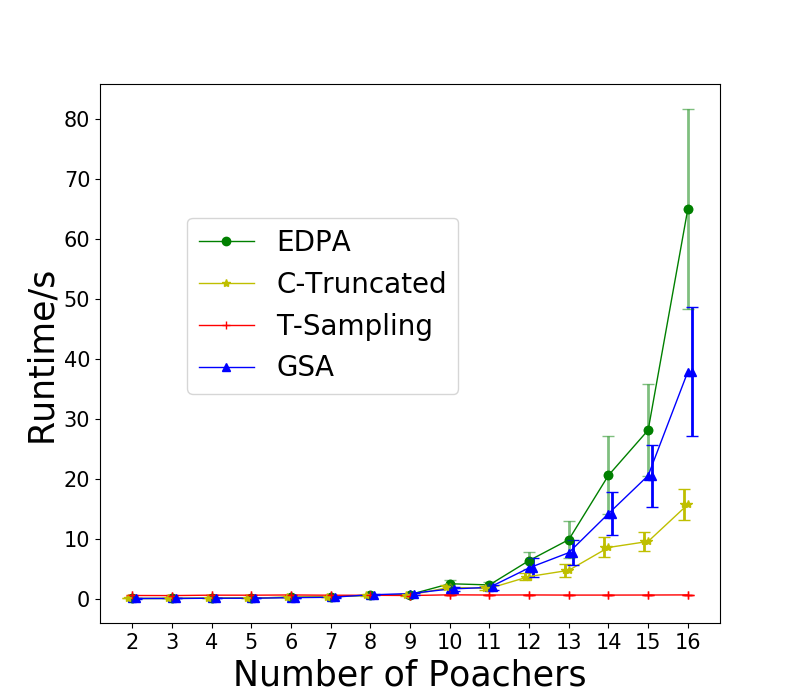}
		\caption{Runtime and solution quality increasing $|Y|$  with $\sum_{p_v}<3.$}
		\label{figbound}
	\end{subfigure}
	~
	\begin{subfigure}[htbp]{0.48\textwidth}
		\centering
		\includegraphics[width=4.0cm]{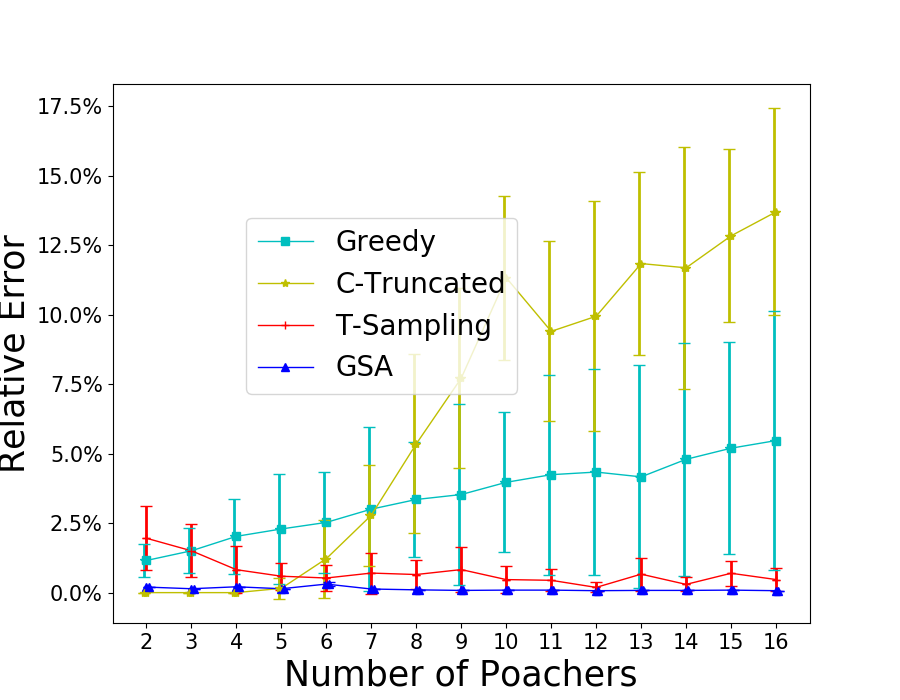}
		~
		\includegraphics[width=3.7cm]{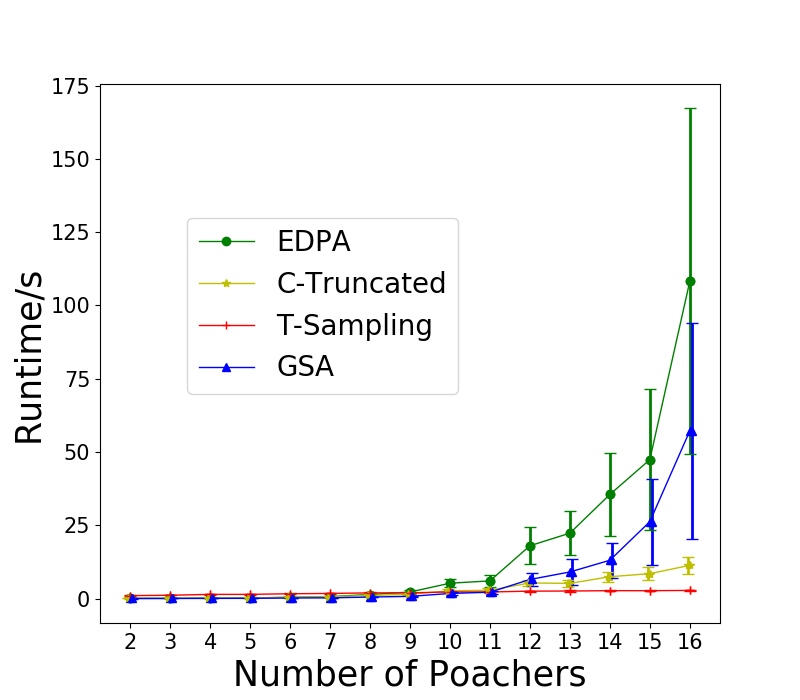}
		\caption{Runtime and solution quality increasing $|Y|$ for General Cases.}
		\label{figunbound}
	\end{subfigure}
	 
	\begin{subfigure}[htbp]{0.48\textwidth}
		\centering
		\includegraphics[width=4cm]{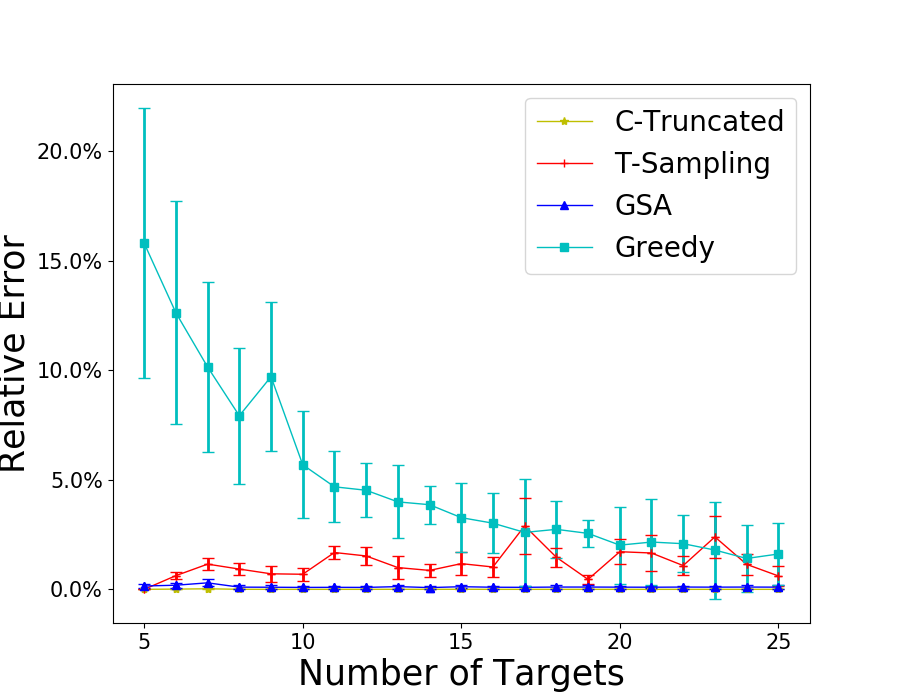}
		~
		\includegraphics[width=3.7cm]{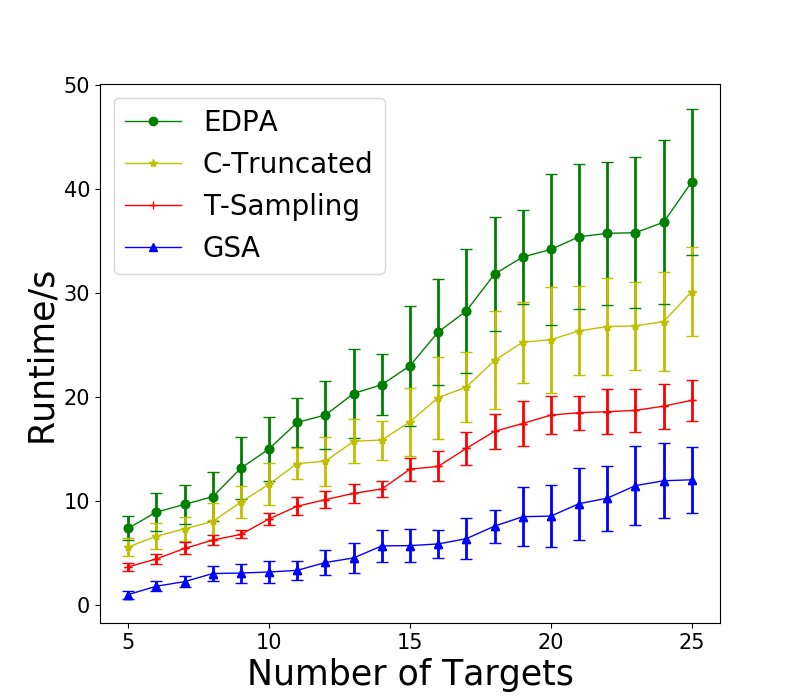}
		\caption{Runtime and solution quality increasing $n$  with $\sum_{p_v}<3.$}
		\label{figtargets}
	\end{subfigure}
	~
	\begin{subfigure}[htbp]{0.48\textwidth}
		\centering
    \includegraphics[width=4.0cm]{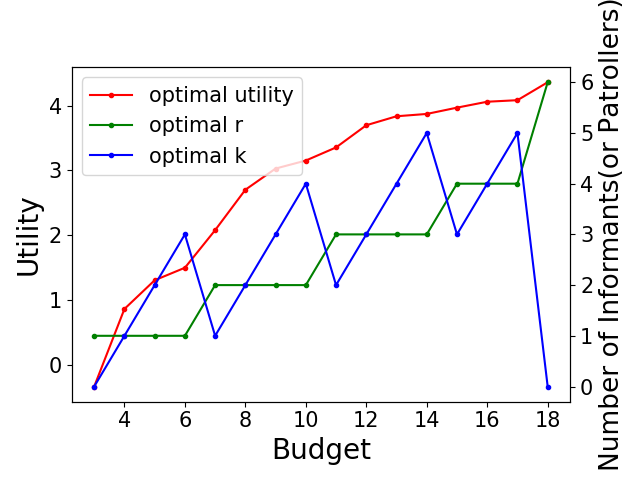}
    ~
    \includegraphics[height=2.7cm]{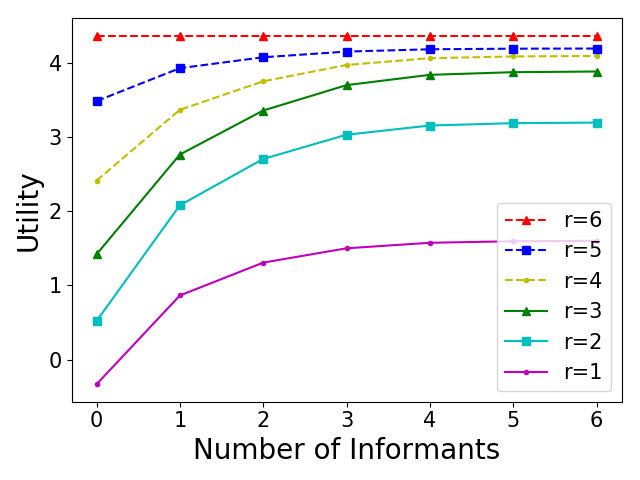}
        \caption{Trade-off between r and k, and increase of utility with fixed $r$ ($|X|=6$, $|Y|=6$, $n=6$).}
        \label{figbudcons}
	\end{subfigure}
	
	\begin{subfigure}[htbp]{0.3\textwidth}
	 \centering
		\includegraphics[height=4.0cm]{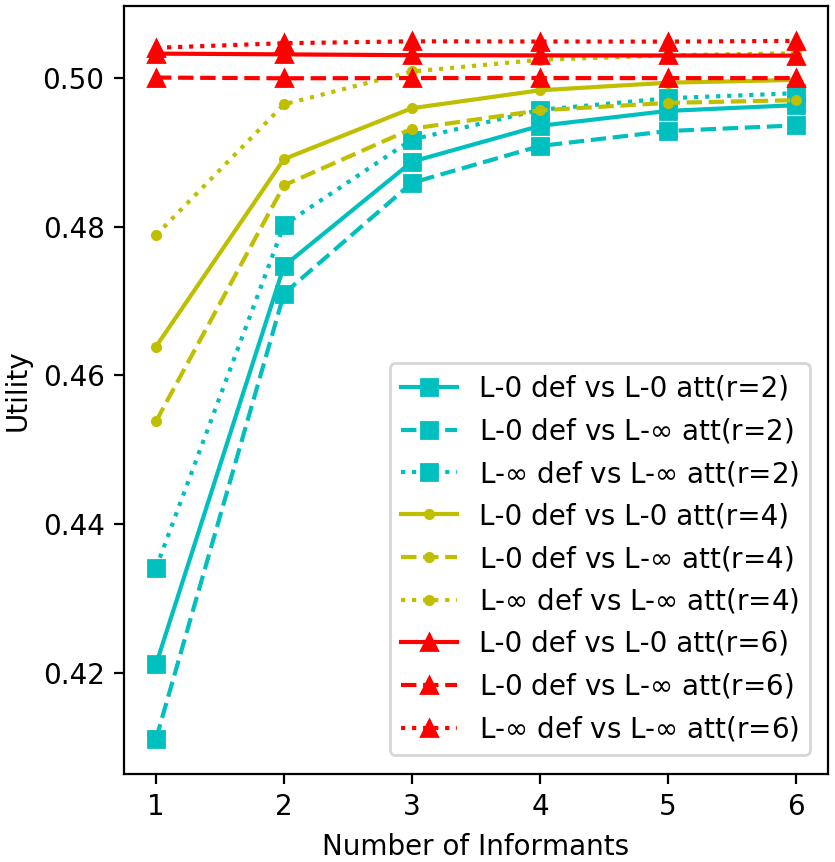}
    	\caption{Comparison between the defender utility against level-0 and level-$\infty$ attackers. ``L-$\infty$/0 def'' means that the defender uses the optimal strategy against a Level-$\infty$/0 attacker.}
    	\label{fig:level0infcomp}
    \end{subfigure}
	\hfill
	\begin{subfigure}[htbp]{0.35\textwidth}
	 \centering
		\includegraphics[height=4cm]{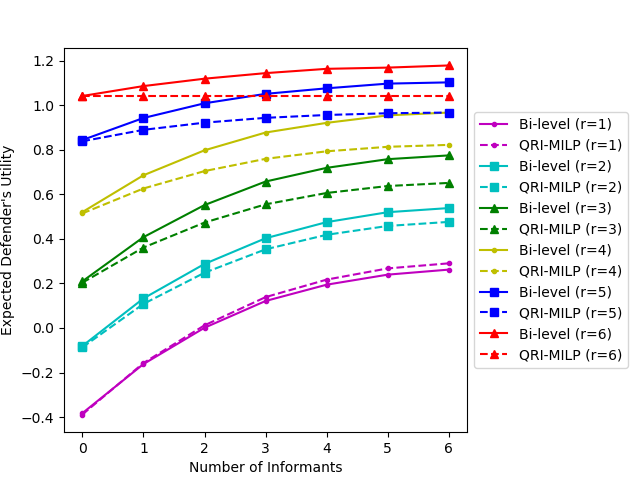}
    	\caption{Comparison between the bi-level optimization algorithm and \texttt{QRI-MILP}.}
    	\label{fig:algo-comp}
    \end{subfigure}
	\hfill
	\begin{subfigure}[htbp]{0.3\textwidth}
	    \centering
		\includegraphics[height=4cm]{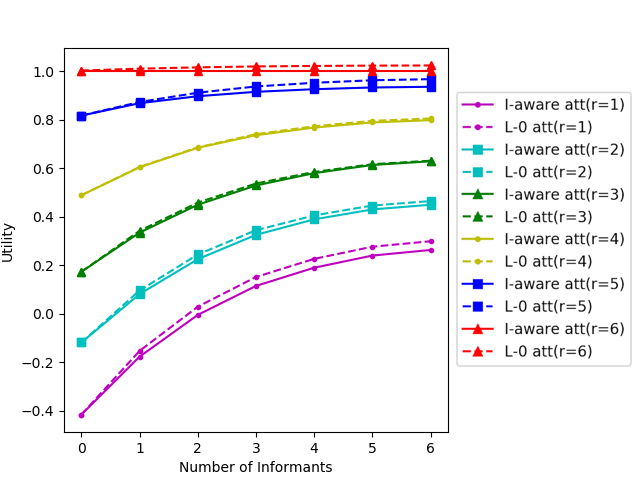}
    	\caption{Comparison between defender utility against level-0 and informant-aware attackers.}
    	\label{fig:adaptive-comp}
	\end{subfigure}
     \caption{Experimental Results.}
\end{figure*}

\subsubsection{Level-0 vs. Level-$\infty$ attackers}
We set $|X|=n=6$, $|Y|=p_1=1$, and $G_S$ to be fully connected. We set $r=2,4,6$ and vary $k$ from $1\ldots6$. 
We first fix the defender's strategy to the one against level-0 attackers and compare the utility achieved by the defender when defending against a level-0 attacker and a level-$\infty$ attacker. 
We show how the defender utility varies with the number of informants and defensive resources in Figure \ref{fig:level0infcomp}.
On average, we see that the defender utility against a level-$\infty$ attacker is lower than that against a level-0 attacker. 
We also show the utility of the defender using her optimal strategy against a level-$\infty$ attacker. We can see that when facing a level-$\infty$ attacker, the defender utility when using the optimal strategy is higher by a margin  than using the one against level-0 attackers.

\subsubsection{Level-0 vs. informant-aware defenders}
\label{sec:adaptvnonadapt}
We set $|X|=n=6$, $|Y|=p_1=1$, and $G_S$ to be fully connected. We vary $r$ from $1\ldots6$ and $k$ from $0\ldots6$.
We assume that the defender recruits the $k$ informants with the highest information sharing intensity $w_{u1}$. The optimal defender strategy against the informant-aware attacker case is found using \texttt{QRI-MILP}. The defender strategy against the level-0 attacker case is computed using \Pasaq~\cite{yang2012computing}. The defender utility against the level-0 attacker is found by first computing $q_i, \DefEU_0$ and then using the results to compute $\DefEU(U)$.

In Figure~\ref{fig:adaptive-comp}, we show how the defender utility in the two cases varies with the number of informants and defensive resources. On average, we see that the defender utility is marginally higher against the level-0 attacker than against the informant-aware attacker, particularly when the defender has either very few or very many defensive resources. 
We also compare the defender's utility of the level-0 defender (defending against level-0 attackers) and the informant-aware defender (defending against informant-aware attackers). The results are deferred to Appendix \ref{app:additionalexp}.

\subsubsection{Comparison between the Bi-Level Algorithm and \texttt{QRI-MILP}}
\label{sec:bi_level_comparison}
We empirically compare the bi-level optimization algorithm with \texttt{QRI-MILP}. 
We set $|X|=n=6$, $|Y|=p_1=1$, and $G_S$ to be fully connected. 
We vary $r$ from $1\ldots6$ and  $k$ from $0\ldots6$.

In both cases, we assume that the defender recruits the $k$ informants with the highest information sharing intensity $w_{u1}$. 
The results are shown in Figure \ref{fig:algo-comp}. In general, our bi-level algorithm gives higher expected defender utilities than the \texttt{QRI-MILP} algorithm, except when $r=1$. Our results show that both increasing the number of resources and hiring more informants increase the defender's utility. However, as the number of resources ($r$) increases, the utility gain from hiring more informants diminishes.

Intuitively, if the number of resources equals the number of targets, the defender should always cover all the targets, Interestingly, during our experiments, we observed that in this case, the optimal defender strategy may not always use all his resources to cover all the targets. The reason is that in a general sum game, by decreasing the probability of protecting a certain target on purpose, the defender can lure the attacker into attacking the target more frequently, and thus increase his expected utility. Such strategies can be found in real-world wildlife protections where the patrollers may sometimes deliberately ignore the tips. This is also reflected in our bi-level algorithm. If the defender always uses all his resources, then both the defender's and the attacker's strategies are fixed, and hiring more informants does not increase the defender's expected utility. But if the defender strategy does not always use all his resources, then hiring more informants could help (see the bi-level algorithm for the $r=6$ case in Figure \ref{fig:algo-comp}). 

\section{Discussion and Conclusion}
\label{sec:conclusion}
In this paper, we introduced a novel two-stage security game model and a multi-level QR behavioral model that incorporated community engagement.
We provided complexity results,
developed algorithms to find (sub-)optimal groups of informants to recruit against level-0 attackers and evaluated the algorithms through extensive experiments. Our results also generalize to the case where informants have heterogeneous recruitment costs and to different kinds of attacker response models, such as SUQR model~\cite{nguyen2013analyzing}, which can be done by calculating the attacker's response correspondingly. Also, see Appendix \ref{app:levelk} on how to extend our algorithms to defend against level-$\kappa$ ($\kappa<\infty$) attackers.
In Section \ref{levelinf}, we defined a more powerful type of attacker that could respond to the marginal strategy and developed a bi-level optimization algorithm to find the optimal defender's strategy in this case.
\gty{slight changes in the beginning}

In the anti-poaching domain, some conservation site managers utilize the so-called ``intelligence'' operations that rely on informants in nearby villages to alert rangers when they know the poachers’ plans in advance. The deployment of the work relies on the site manager to provide their understanding of the social connections among community members. The edges and parameters of the bipartite graph in our model can be extracted from a local social media application or historical data collected by site managers. 
Recruiting and training reliable informants is costly and managers may only afford a limited number of them. Our model and solution can help the managers efficiently recruit informants, make the best use of tips and evaluate the trade-off between allocating budget to hiring rangers and recruiting informants in a timely fashion.

\gty{slight change in the beginning}
For future work, instead of using a particular behavior model, we can use historical records as training data and learn the attackers' behavior in different domains. It would also be interesting to consider the case where the informants can only provide inaccurate tips or other types of tips, e.g., some subset of targets will be attacked instead of a single location. We can also model the informants as strategic agents. In real life, it is possible that informants may also provide fake information if they have their own utility structures. We can try to reward them to elicit true information and maximize the defender's utility.

\section*{Acknowledgement}
This work is supported in part by NSF grant IIS-1850477 and a research grant from Lockheed Martin.

\balance
\bibliographystyle{ACM-Reference-Format}  
\bibliography{main}  
\clearpage
\nobalance
\appendix
\section*{ \LARGE Appendix:\\ Green Security Game with Community Engagement}
\section{Complexity of different algorithms}
\label{app:complexity-table}
\begin{table}[!h]
\centering
\begin{tabular}{lr}  
\toprule
Algorithm  & Time Complexity  \\
\midrule
\BF & $O(|X|^k2^{|Y|}n^2r|Y|^3)$\\

\Ctrunc & $O(|X|^kn^2r|Y|^{C+3})$\\

\Sampling & $O(|X|^k\mathsf{T}n^2r|Y|^3)$\\

\SpeCase & $O(|X|^k n^2r|Y|^4)$\\

\Heuristic & $O(2^{k}|X|n^2r|Y|^3)$\\
\bottomrule
\end{tabular}
\caption{Complexity Table}
\label{table2}
\end{table}

\section{Proof of Theorem \ref{thmalloc}}\label{app:thmalloc}
\thmalloc*
\begin{proof}
Given the tips $\mathbf{V}$, the defender should calculate $\EG_i(|V_i|,V_0)$ for each target $i\in T$, and then allocate the resources to $r$ of the targets with the highest $\EG_i$.

The above strategy is indeed optimal since the expected utility with no resources is given by $\sum_{i\in T}\EU^u_{i}(|V_i|,V_0)$, and once an additional unit of resource is given, it should always be allocated to the uncovered target that could lead to the largest increment in expected utility, i.e., the target with the largest $\EG_i(|V_i|,V_0)$. 

The calculation of $\EG_i(|V_i|,V_0)$ for each $i\in T$ can be done in $O(n+|Y|)$ time, and finding the $r$ largest $\EG_i(|V_i|,V_0)$ can be done in $O(n)$ time, leading to the overall complexity of $O(|Y|+n)$.
\end{proof}

\section{Proof of Theorem~\ref{thm:nphard}}
\label{app:nphard}
\thmnphard*
\begin{proof}
Consider the case where $r=1$, $p_v=w_{uv}=1$ for all $u,v$, and the targets are uniform i.e., $R_i^d$'s ($R_i^a, P_i^d, P_i^a$) are the same for all $i\in T$. We use the notation $R^d$ ($P^d$) instead of $R_i^d$ ($P_i^d$) for simplicity. Let $\lambda=0$.
	
    To start with, we investigate how $\DefEU(U)$ depends on a given $U$. Since $p_v=1$ and $w_{uv}=1$ for all $u\in X,v\in Y$, all attackers in $V$ will be reported to attack a location. Let random variable $X_i=|V_i|$ be the number of attackers who are reported to attack location $i$. 
    Since the targets are uniform, an attacker will attack each location with  probability $q_i=\frac{1}{n}$ if he goes attacking.
	Then the defender's expected utility  $\DefEU(U)$ could be written as 
	\begin{align*}
	\DefEU(U)
		=&\left(\mathbb{E}\left[\max_{i\in T}X_i\right]+\frac{|Y|-|V|}{n}\right)R^d\\
		&+\left(|V|-\mathbb{E}\left[\max_{i\in T}X_i\right]+\frac{n-1}{n}(|Y|-|V|)\right)P^d\\
		=&\left(\mathbb{E}\left[\max_{i\in T}X_i\right]-\frac{|V|}{n}\right)\left(R^d-P^d\right)+\frac{|Y|}{n}R^d+\frac{(n-1)|Y|}{n}P^d.
	\end{align*}
	The latter two terms are independent of the choice of informants so to maximize $\DefEU$, it suffices the maximize $\left(\mathbb{E}\left[\max_{i\in T}X_i\right]-\frac{|V|}{n}\right)$.
	
	We can prove by induction on $|V|$ that $\left(\mathbb{E}\left[\max_{i\in T}X_i\right]-\frac{|V|}{n}\right)$ increases as $|V|$ increases, or $\mathbb{E}\left[\max_{i\in T}X_i\right]$ increases by at least $\frac{1}{n}$ if $|V|$ is increased by 1: 
	\begin{enumerate}
	    \item  Since $\mathbb{E}\left[\max_{i\in T}X_i\right]=1$ when $|V|=1$ and $\mathbb{E}\left[\max_{i\in T}X_i\right]=1+\frac{1}{n}$ when $|V|=2$, it holds for $|V|=1$. 
        \item Consider $|V|\geq 1$ and the corresponding sequence $\{X_i\}_{i=1}^n$. Let $X_m=\max_{1\leq i\leq n}\{X_i\}$. We add an attacker to $V$ and denote by $p$ the probability of he targeting the location with the largest $X_i$. Thus the expected maximum increase by $p$. Since $p\geq \frac{1}{n}$ and by a simple coupling argument, we have that $\mathbb{E}\left[\max_{i\in T}X_i\right]$ increases by at least $\frac{1}{n}.$
	\end{enumerate}

	Thus, in this case, solving for the optimal solution of the original problem is equivalent to solving for $U$ that maximizes the size of $V$ in the first stage.

	We show that the optimization problem is NP-Hard using a reduction from \MCP: we are given a number $k$ and a collection of sets $S$. The objective is to find a subset $S'\subseteq S$ of sets such that $|S'|\leq k$ and the number of covered elements $\left|\bigcup_{S_i\in S'}S_i\right|$ is maximized.
	Let $X=\{x_1,\ldots,x_{|S|}\}$, $Y=\bigcup_{S_i\in S}S_i$, $E=\{(x_i,y):i\in[|S|]\land y\in S_i\}$, $p_v=1$ for all $v\in Y$ and $W_e=1$ for all $e\in E$. Thus to find a $U\subseteq X$ with $|U|\leq k$ that maximizes the size of $V$ is equivalent to finding a subset of sets with size no larger than $k$ that maximizes the number of covered elements in the instance of \MCP.
	\end{proof}
	
	\section{Proof of Lemma \ref{lem:SISI}}
	\label{app:prooflemSISI}
	\lemSISI*
	\begin{proof}
	Since $w_{uv}=1$ for all $u,v$, we have $\tilde{p}_v(V_0)=0$ for each $v\in V\setminus V_0$ given $V_0$. Therefore, the expected gain of target $j$ with $\tilde{y}$ reported attacks can be written as  $\EG_j=(\tilde{y}+q_{j}\sum_{v\in Y\setminus V}p_v) (R_j^d-P_j^d)$ and the calculation of $f(\cdot)$ depends only on the size of $|V_0|$. Thus, instead of enumerating $V_0$, we enumerate $0\leq t_0\leq |V|$ as the size of $V_0$ in line 2 of Algorithm \ref{AlgoBF}, and replace $P_{V_0}$ in Algorithm \ref{AlgoBF} with $P_{t_0}$, where $P_{t_0}=\Pr[|V_0|=t_0|U]$ can be obtained by expanding the following polynomial $\prod_{v\in V}(1-p_v+p_vx)=\sum_{i=0}^{|V|}P_{i}x^i.$
	Therefore, $\DefEU(U)$ can be calculated in $O(n^2r|Y|^4)$ time.
	
	Since all possible $U$ can be enumerated in $O(|X|^k)$, the optimal set of informants can be computed in $O(|X|^kn^2r|Y|^4)$.
	\end{proof}

\section{ASISI}\label{app:asisi}
In Algorithm~\ref{AlgoBFSISI} and Algorithm~\ref{AlgoDPSISI}, we present the polynomial time algorithm used to compute $\DefEU(U)$ when $w_{uv} = 1, \forall (u, v) \in E$.

\begin{algorithm}[h]
	\caption{Calculate $\DefEU(U)$}\label{AlgoBFSISI}
	\begin{algorithmic}[1]
		\State Expand polynomial $\prod_{v\in V}(1-p_v+p_vx)=\sum_{j=0}^{|V|}P_{j}x^j$
		\State $\EU\gets 0$
		\For{all possible number of reported attackers $0\leq t_0\leq |V|$}
		\If {$t_0=0$}
		\State $\EU=\EU+P_{t_0}\sum_{v\in Y}p_v\DefEU_0$
		\State Continue to line 2
		\EndIf
		\For {target $i\in T$ and $0\leq t_i\leq t_0$} \Comment{Enumerate target $i$ and the number of attackers $t_i$ targeting $i$}
		
		\State Calculate $f(\cdot)$ given $t_0,i,t_i$
		\State $\EG_i\gets(t_i+q_i\sum_{v\in Y\setminus V}p_v) (R_i^d-P_i^d)$
		\State $\EU_i^u\gets (t_i+q_i\sum_{v\in Y\setminus V}p_v)P_i^d$
		\State $P_{i,r}\gets (t_0-t_i)! \left(\sum_{x=0}^{r-1}f(s,x,t_0-t_i)\right)$
		\State $\EU=\EU+P_{t_0} {t_0 \choose t_i}q_i^{t_i}\cdot P_{i,r}\cdot EG_i$
		\State $\EU=\EU+P_{t_0} {t_0 \choose t_i}q_i^{t_i}(1-q_i^{t_i})\EU_i^u$
		
		\EndFor
		\EndFor
		\State $\DefEU(U)\gets \EU$

	\end{algorithmic}
\end{algorithm}

\begin{algorithm}[h]
    \caption{Calculate $f(\cdot)$ given $t_0$,$i$,$t_i$}\label{AlgoDPSISI}
    \begin{algorithmic}[1]
    \State $\{{i_1},\ldots,{i_{n-1}}\}\gets T\setminus \{i\}$
    \State $\EG_{i}=(t_i+q_{i}\sum_{v\in Y\setminus V}p_v) (R_i^d-P_i^d)$
    \State Initialize $f(s,x,y)\gets 0$ for all $s,x,y$ 
		\State $f(0,0,0)\gets 1$
		\For{$s\gets 1$ to $n-1$}
		\For{$x\gets 0$ to $\min(s,r)$}
		\For{$y\gets 0$ to $t_0-t_i$}
		\For{$\tilde{y}\gets 0$ to $y$}
		\State $\EG_{i_s}=(\tilde{y}+q_{i_s}\sum_{v\in Y\setminus V}p_v) (R_i^d-P_i^d)$
		\If {$\EG_{i_s}>\EG_i$}
		\State $f(s,x,y)\gets f(s,x,y)+\frac{q_{i_s}^{\tilde{y}}}{\tilde{y}!} f(s-1,x-1,y-\tilde{y})$
		\Else 
		\State $f(s,x,y)\gets f(s,x,y)+\frac{q_{i_s}^{\tilde{y}}}{\tilde{y}!}f(s-1,x,y-\tilde{y})$
		\EndIf
		\EndFor  
		\EndFor
		\EndFor
		\EndFor
    \end{algorithmic}
\end{algorithm}

\section{Defender Utility is not Submodular}
\label{app:submodular}
 We provide a counterexample that disproves the submodularity of $\DefEU(U)$.

\begin{example}
Consider a  network $G_S=(X,Y,E)$ where $X=\{u_1,u_2\}$, $Y=\{v_1,v_2,v_3\}$, $E=\{(u_1,v_2),(u_2,v_3)\}$, $p_v=1 \,\forall v\in Y$ and $w_{uv}=1\, \forall (u,v)\in E$. There are 2 targets $T=\{1,2\}$, where $R_i^d=i$, $P_i^d=-10^{-8}\approx 0$ for any $i\in T$. 
Letting $\lambda=0$ yields $q_i=0.5$. The defender has only 1 resource. We can see that $\DefEU(\emptyset)=\DefEU(\{1\})=\DefEU(\{2\})=3$, $\DefEU(\{1,2\})=\frac{1}{4}(2+0.5)+\frac{1}{4}(4+1)+\frac{1}{2}(2+1)=3.375.$ As a result, $\DefEU(\{1,2\})+\DefEU(\emptyset)>\DefEU(\{1\})+\DefEU(\{2\}).$

\end{example}


\section{Proof of Lemma~\ref{lemmabound}}
\label{app:ctruncated-proof}
\lemmabound*
\begin{proof}
Let the random variable $W$ be the number of attacks. Let $\mathcal{A}_1$ be the set of events of having  no less than $C$ reported attackers   and $\mathcal{A}_2$ be the set of events of having  no less than $C$ attacks. Let $E_A$ be the expected defender's utility  taken over all possible tips given an event $A$.
By noticing that $\mathcal{A}_1\subseteq\mathcal{A}_2$, we have 
\begin{eqnarray}
&&|\DefEU(U)-{\DefEU}(U,C)|\nonumber \\
&\leq& \sum_{A\in \mathcal{A}_1}\Pr[A]|E_A|\leq \sum_{A\in \mathcal{A}_2}\Pr[A]|E_A|\nonumber\\
&\leq& Q\sum_{i=C}^{|Y|}\Pr[W= i]\cdot i\nonumber\\
&=&Q\left(C\Pr[W\geq C]+\sum_{i= C+1}^{|Y|}\Pr[W\geq i]\right)\nonumber\\
&\leq& Q\left(Ce^{-2(C-C')^2/|Y|}+\sum_{i=C+1}^{|Y|}e^{-2(i-C')^2/|Y|}\right)\label{eq:chernoff}\\
&<& Q \cdot e^{-2(C-C')^2/|Y|}\left( C+\frac{1}{1-e^{-4(C-C')/|Y|}}\right).\label{eq:final_bound}
\end{eqnarray}
Inequality \eqref{eq:chernoff} follows by the Chernoff Bound, and inequality \eqref{eq:final_bound} follows since 
\begin{eqnarray*}
	&&\sum_{i=C+1}^{|Y|}e^{-2(i-C')^2/|Y|}\\
	&\leq& e^{-2(C+1-C')^2/|Y|}\sum_{i\geq 0}e^{-4i(C+1-C')/|Y|}\\
	&<& e^{-2(C-C')^2/|Y|}\cdot \frac{1}{1-e^{-4(C-C')/|Y|}}.
\end{eqnarray*}
\end{proof}

\section{Non-convergence of the level-$\kappa$ response}
\label{app:counter-example}
\begin{example}
\label{example:non_convergence}
Suppose there is a single attacker and two targets with the following payoffs:
\begin{gather*}
    R^a_1 = 0.6, R^a_2 = 0.8,\\
    P^a_1 = -0.8, P^a_2 = -0.6.
\end{gather*}
In this case, there are only two possible tips: the attacker attacks target 1 ($\mathbf{V}_1$), and the attacker attacks 2 ($\mathbf{V}_2$). Assume that only 1 informant is recruited with report probability $w=0.5$. The defender has only 1 defensive resource and uses the following strategy:
\begin{gather*}
    \mathbf{x}_0=(0.5, 0.5), \mathbf{x}(\mathbf{V}_1)=(1.0, 0.0), \mathbf{x}(\mathbf{V}_2)=(0.0, 1.0).
\end{gather*}
When the attacker has $\lambda=2.9$, the level-$\kappa$ response will converge to $\mathbf{q}=(0.4283, 0.5717)$. However, if $\lambda=3.0$, then the process will eventually oscillate between $\mathbf{q}=(0.2924, 0.7076)$ and $\mathbf{q}'=(0.5676, 0.4324)$ iteratively.
\end{example}

\section{Proof of Lemma \ref{lem:lipschitz}}
\label{app:ipschitz-proof}
\lemLipschitz*
\begin{proof}
    Given defender's strategy $\mathbf{x}_0$ and $\mathbf{x}(\mathbf{V})$, define:
    \begin{gather*}
        g(\mathbf{q})=\QR(\MS(\mathbf{x}_0, \mathbf{x}, \mathbf{q})).
    \end{gather*}
    Then a level-($\kappa$+1) attacker's strategy can be computed by
    \begin{gather*}
        \mathbf{q}^{\kappa+1}=g(\mathbf{q}^{\kappa})=g(g(\mathbf{q}^{\kappa-1}))=\cdots=g^{\kappa}(\mathbf{q}).
    \end{gather*}
    The convergence of level-$\kappa$ is equivalent to the convergence of $g^{\kappa}(\mathbf{q})$.

	The marginal strategy $\hat{\mathbf{x}}$ can be written as:
	\begin{gather*}
		\hat{\mathbf{x}}(\mathbf{q})=\sum_{\mathbf{V}}\Pr\{\mathbf{V}\}\mathbf{x}(\mathbf{V})=(1-w)\mathbf{x}_0+\sum_{i}wq_i\mathbf{x}(\mathbf{V}_i).
	\end{gather*}
	Notice that the function $g(\mathbf{q})$ is just the quantal response against $\hat{\mathbf{x}}$:
	\begin{gather*}
		g_i(\mathbf{q}) = \frac{e^{\lambda u^a_i(\hat{x}_i)}}{\sum_{j}e^{\lambda u^a_j(\hat{x}_j)}},
	\end{gather*}
	where $u^a_i(\hat{x}_i)$ is the attacker's expected utility of attacking target $i$ when the defender's marginal strategy is $\hat{\mathbf{x}}$: $u^a_i(\hat{x}_i) = R^a_i-\hat{x}_i(R^a_i-P^a_i)$.
	Therefore,
	\begin{align*}
		\frac{\partial g_i}{\partial q_j}=&\frac{\partial g_i}{\partial u^a_i}\frac{\partial u^a_l}{\partial q_j}+\sum_{l\ne i}\frac{\partial g_i}{\partial u^a_l}\frac{\partial u^a_l}{\partial q_j} \\
		=&\lambda wg_i(1-g_i)(P^a_i-R^a_i)x_i(\mathbf{V}_j)+\sum_{l\ne i}\lambda wg_ig_l(R^a_l-P^a_l)x_k(\mathbf{V}_j)\\
		=&\lambda wg_i\left[(P^a_i-R^a_i)x_i(\mathbf{V}_j)+\sum_{l}g_l(R^a_l-P^a_l)x_l(\mathbf{V}_j) \right].
	\end{align*}
	Note that in the above equation, $0\le w, g_i,g_l,x_i(\mathbf{V}_j),x_l(\mathbf{V}_j)\le 1$, $P^a_i-R^a_i<0$ and $R^a_l-P^a_l>0$. Thus we have:
	\begin{gather}
		\label{eq:direvative}
		\lambda (P^a_i-R^a_i)x_i(\mathbf{V}_i)<\frac{\partial g_i}{\partial q_j}<\lambda \sum_{l}g_l(R^a_l-P^a_l)x_l(\mathbf{V}_j).
	\end{gather}
	
	On the other hand, $\bar{x}_i\le \frac{L}{n\lambda(R^a_i-P^a_i)}$ means $x_i(\mathbf{V}_j)\le \frac{L}{n\lambda(R^a_i-P^a_i)},\forall j$. Plugging into Equation \eqref{eq:direvative}, we get:
	\begin{gather*}
		-\frac{L}{n}<\frac{\partial g_i}{\partial q_j}<\frac{L}{n}\sum_{l}g_l=\frac{L}{n}.
	\end{gather*}
	For any $\mathbf{q}'\ne \mathbf{q}$, let $\mathbf{q}^{(i)}=(q_1,\dots,q_i,q'_i,\dots,q'_n)$. So $\mathbf{q}^{(0)}=\mathbf{q}'$ and $\mathbf{q}^{(n)}=\mathbf{q}$. Therefore,
	\begin{align*}
	&\left\|g(\mathbf{q})-g(\mathbf{q}')\right\|_1
	=\sum_{i=1}^n\left|g_i(\mathbf{q})-g_i(\mathbf{q}')\right|\\
	=&\sum_{i=1}^n\left|\sum_{j=1}^{n}\left[g_i(\mathbf{q}^{(j)})-g_i(\mathbf{q}^{(j-1)})\right]\right|\\
	=&\sum_{i=1}^n\sum_{j=1}^{n}\left|g_i(\mathbf{q}^{(j)})-g_i(\mathbf{q}^{(j-1)})\right|\\
	=&\sum_{i=1}^n\sum_{j=1}^{n}\left|\int_{q'_j}^{q_j}\left.\frac{\partial g_i(\mathbf{q})}{\partial q_j}\right|_{\mathbf{q}=(q_1,\dots,q_{i-1},s,q'_{i+1},\dots,q'_n)}\,\mathrm{d}s\right|\\
	<&\sum_{i=1}^n\sum_{j=1}^{n}\int_{q'_j}^{q_j}\frac{L}{n}\,\mathrm{d}s
	\le\sum_{i=1}^n\sum_{j=1}^{n}\frac{L}{n}|q_j-q'_j|\\
	=&\sum_{i=1}^n\frac{L}{n}\|\mathbf{q}-\mathbf{q}'\|_1
	=L\|\mathbf{q}-\mathbf{q}'\|_1.
	\end{align*}
\end{proof}

\section{Algorithm for Defending against Informant-Aware Attackers}
\label{app:infawareatt}

\begin{restatable}[]{proposition}{propositionQRI}
\label{non-decreasing}
The optimal objective of \texttt{QRI} is non-decreasing in $w$.
\end{restatable}

\begin{proof}
Consider the two optimization problems induced by different values for $w$: $w_1, w_2$ where $w_2 > w_1$. Let $(\mathbf{x}, \mathbf{y}, \mathbf{z})$ be a solution for when $w = w_1$. Then, $(\mathbf{x}, \mathbf{y}, \frac{w_1}{w_2}\mathbf{z} + (1 - \frac{w_1}{w_2})\mathbf{x})$ is a feasible solution for when $w = w_2$ that achieves the same objective value. To see why it is feasible, observe that constraint (\ref{optln:y}) is satisfied by construction and constraint (\ref{optln:z}) is satisfied since the new value for each $z_i$ is a convex combination of the previous $x_i, z_i$, which were both in $[0, 1]$.
\end{proof}

 Proposition~\ref{non-decreasing} implies that when selecting informants, it is optimal to simply maximize $w$. Since $w = 1 - \prod_{u \in U} (1 - w_{u1})$, we can select informants greedily and choose the $k$ informants with the largest information sharing intensity $w_{u1}$. We can then solve the optimization problem to find the optimal allocation of resources. 
Finally, we discuss how to find an approximate solution to \texttt{QRI} using a MILP approach.

\label{sec:milp-solution}
We can compute the optimal defender strategy by adapting the approach used in the \Pasaq algorithm~\cite{yang2012computing}. Let $N(y), D(y)$ be the numerator and denominator of the objective in \texttt{QRI}. As with \Pasaq, we binary search on the optimal value $\delta^*$. We can check for feasibility of a given $\delta$ by rewriting the objective to $\min_{x, y, z}\delta D(y) - N(y)$ and checking if the optimal value is less than $0$. To solve the new optimization problem, which still has a non-linear objective function, we adapt their approach of approximating the objective function with linear constraints and write a MILP.

First, let $\theta_i = e^{\lambda R^a_i}$, $\beta_i = \lambda(R^a_i - P^a_i)$, and $\alpha_i = R^d_i - P^d_i$. We rewrite the objective as

$$\sum_{i \in \calT}\theta_i(\delta - P^d_i)e^{-\beta_i(y_i)} + \sum_{i \in \calT}\theta_i\alpha_iy_ie^{-\beta_i(y_i)}$$

\begin{figure}[h]
    \centering
    \includegraphics[width=0.45\textwidth]{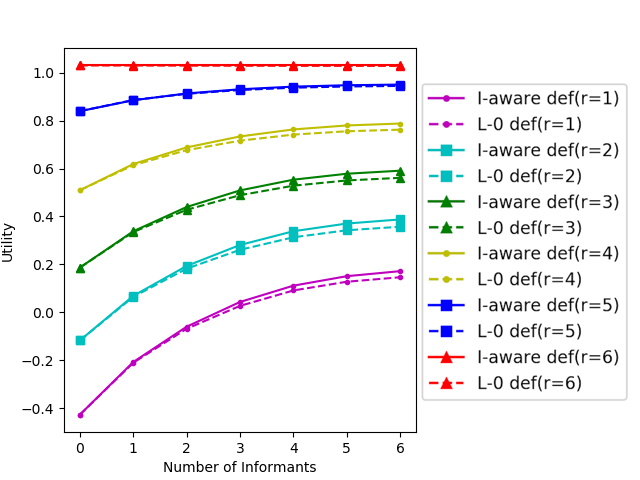}
    \caption{Utility comparison between the level-0 defender and the informant-aware defender against an informant-aware attacker.}
    \label{fig:simple-optimal}
\end{figure}

We have two non-linear functions that need approximation $f^1_i(y) = e^{-\beta_i y}$ and $f^2_i(y) = ye^{-\beta_i y}$. Let $\gamma_{ij}$ be the slope for the linear approximation of $f^1_i(y)$ from $(\frac{j}{K}, f^1_i(\frac{j}{K}))$ to $(\frac{j + 1}{K}, f^1_i(\frac{j + 1}{K}))$ and similarly with $\mu_{ij}$ for $f^2_i(y)$.

The key change in our MILP compared to \Pasaq is that we replace the original defender resource constraint with constraints (\ref{optln:change-start}) - (\ref{optln:change-end}), which take into account the ability of the defender to respond to tips.

\texttt{QRI-MILP}:
\begin{align}
    \min_{x, y, z, a}  \quad \sum_{i \in \calT}&\theta_i(\delta - P^d_i)(1 + \sum_{j = 1}^{K}\gamma_{ij}y_{ij}) - \sum_{i \in \calT}\theta_i\alpha_i\sum_{j=1}^{K}\mu_{ij}y_{ij} \nonumber\\
    \text{subject to} & \quad \sum_{j = 1}^{K} y_{ij} = (1 - w)x_i + wz_i, \quad \forall i \label{optln:change-start}\\
    & \quad \sum_{i \in \calT} x_i \leq r\\
    & \quad 0 \leq x_i \leq 1, \quad \forall i \\
    & \quad 0 \leq z_i \leq 1, \quad \forall i \label{optln:change-end}\\
    & \quad 0 \leq y_{ij} \leq \frac{1}{K}, \quad \forall i, j = 1 \ldots K \label{optln:milp-start}\\
    & \quad a_{ij}\frac{1}{K} \leq y_{ij}, \quad \forall i, j = 1 \ldots K-1\\
    & \quad y_{i(j+1)} \leq a_{ij}, \quad \forall i, j = 1 \ldots K-1\\
    & \quad a_{ij} \in \{0, 1\}, \quad \forall i, j = 1 \ldots K-1\label{optln:milp-end}
\end{align}

\begin{proposition}
The feasible region for $\mathbf{y} = \langle y_i = \sum_{j=1}^K y_{ij}, i \in \calT \rangle$ of \texttt{QRI-MILP} is equivalent to that of \texttt{QRI}.
\end{proposition}
\begin{proof}
With the substitution $y_i = \sum_{j=1}^K y_{ij}$, constraints (\ref{optln:change-start}) - (\ref{optln:change-end}) are directly translated from \texttt{QRI}. The remaining constraints (\ref{optln:milp-start}) - (\ref{optln:milp-end}) can be shown to allow for any potential $y_i$, represented correctly with the appropriate $y_{ij}$.
\end{proof}

With the above claim shown, the proof for the approximate correctness of \Pasaq applies here and we can show that we can find an $\varepsilon$-optimal solution for arbitrarily small $\varepsilon$~\cite{yang2012computing}.

\section{Additional Experiment Results}
\label{app:additionalexp}
In Figure~\ref{fig:simple-optimal}, we compare the performance of the level-0 defender and the informant-aware defender when playing against an informant-aware attacker. We see that despite that their strategies are computed under the level-0 attacker assumption, the utility of the level-0 defender is only slightly lower than the utility of the informant-aware defender. For a fixed $r$, the difference in utility grows larger as the defender recruits more informants and has a higher probability of receiving a tip.

\begin{figure}[htbp]
		\centering
		\includegraphics[width=0.24\textwidth]{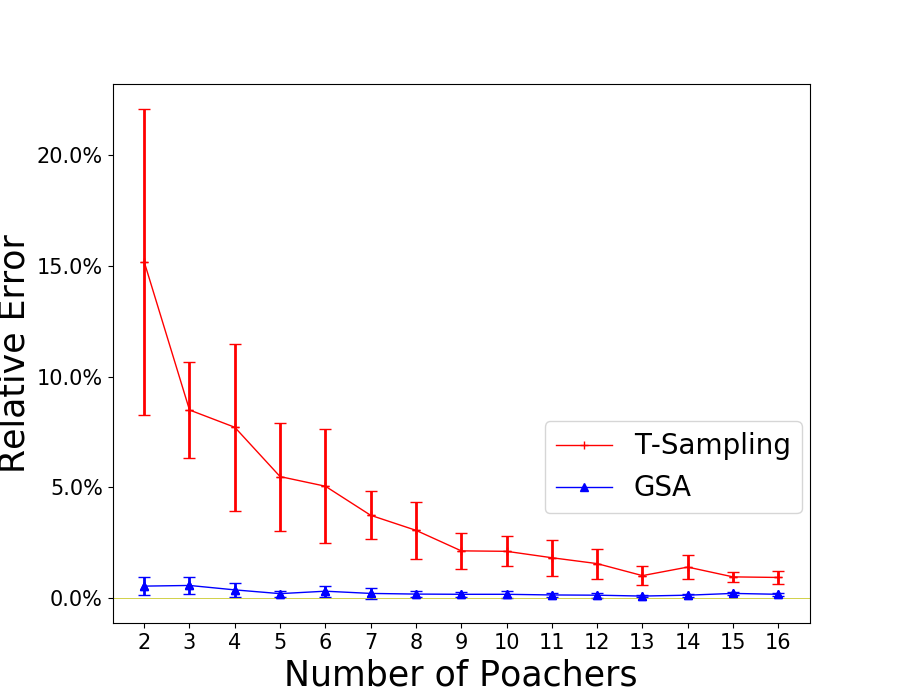}
		~
		\includegraphics[width=0.24\textwidth]{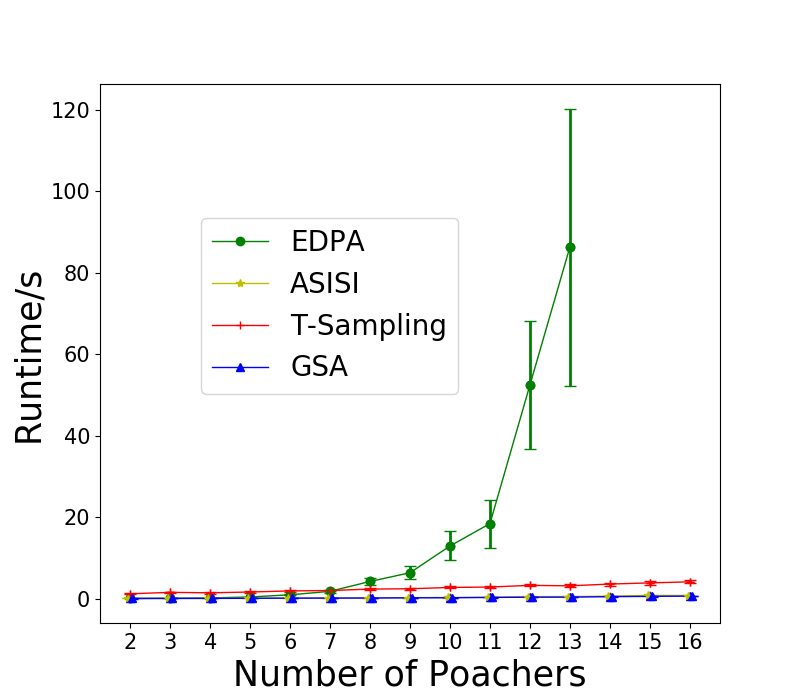}
		\caption{Runtime and Solution Quality increasing $|Y|$  with $w_{uv}=1$ for all $u,v$.}
		\label{figspec}
	\end{figure}

	\begin{figure}[htbp]
		\centering
		\includegraphics[width=4.0cm]{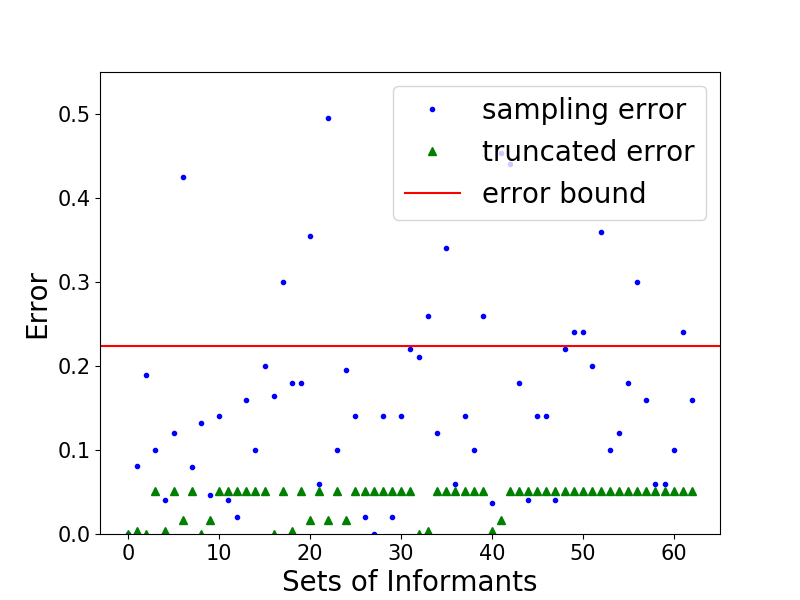}
		~
		\includegraphics[width=4.0cm]{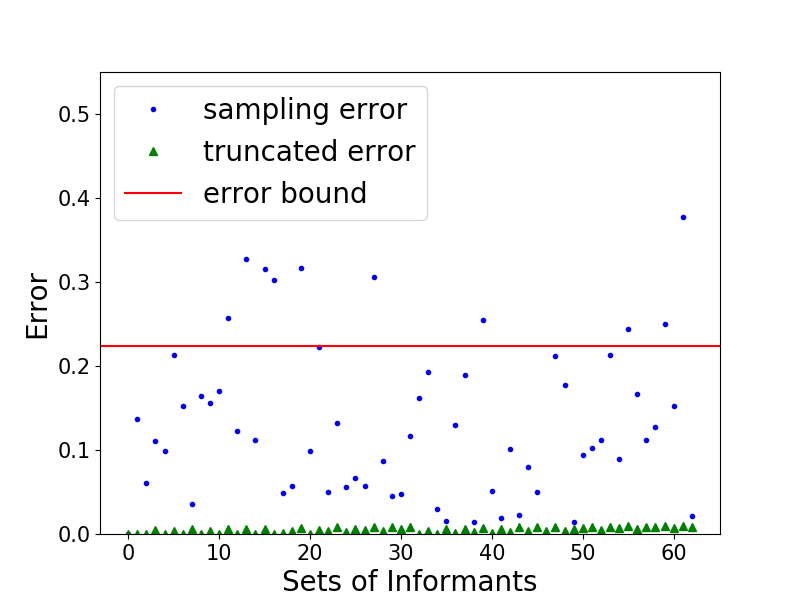}

		\caption{Error of $\DefEU$ on 2 Cases with Fixed $C=6,C'=2,|Y|=8,Q=2$.}
		\label{figerr}
	\end{figure}

We test the special case assuming SISI. We set $|X|=6, k=4, n=10, r=3$ and enumerate $|Y|$ from 2 to 16. 
The average runtime of all algorithms including \SpeCase together with the relative error of \Heuristic and \Sampling are shown in Figure \ref{figspec}. Though \SpeCase and \Heuristic are the fastest among all, the average relative errors of \Heuristic are slightly above $0\%$. \Sampling is slightly slower than \SpeCase and \Heuristic, and the solutions are less accurate than the other two in this case.

Another experiment is a case study on 2 instances with $\sum_{v\in Y}p_v<2$ fixed $|X|=6$, $|Y|=8$, $n=6$, $r=3$, $Q=2$. We run \BF, \Ctrunc ($C=6$) and \Sampling on each instance and show the error of the estimations for all $U\subseteq X$. The results are shown in Figure \ref{figerr} and the red lines indicate the error bound given by Lemma \ref{lemmabound}. We encode the set of informants in binary, e.g., the set with code $19=(010011)_2$ represents the set $\{u_1,u_2,u_4\}.$
The first instance is constructed to show that the bound given by Lemma \ref{lemmabound} is empirically tight, i.e., the estimation of $\DefEU(U)$ by \Ctrunc could be large but still bounded. In this case, we set $p_v=1, w_{uv}=1\, \forall u,v$, $R_i^d=Q$ and $P_i^d=-10^{-3}.$ While the other instance is randomly generated. It is shown that \Sampling has larger errors with higher variances compared to \Ctrunc.

\section{Experiment Setup}
\label{app:expsetup}
To generate $G_S=(X,Y,E)$, we first fix the sets $X$ and $Y$. For each $u\in X$, we sample the degree of $u$, $d_u$, uniformly from $[|Y|]$ and then sample a uniformly random subset of $Y$ of size $d_u$. For each $(u,v) \in E$, $w_{uv}$ is drawn from $U[0,0.2]$.
For the attack probability $p_v$, in the general case, each $p_v$ is drawn from $U[0.4,1]$.
When we restrict $\sum_{v\in Y}p_v\leq C'$, we draw a vector $\mathbf{t}=(t_1,\ldots,t_{|Y|})$ from $U[0,1]^{|Y|}$ and set $p_v=\min\{1,C'\cdot\frac{t_v}{||\mathbf{t}||_1}$\}. 
For the payoff matrix, each $R_i^d$ ($R_i^a$) is drawn from $U(0,Q]$ and each  $P_i^d (P_i^a)$ is drawn from $U[-Q,0)$, where $Q$ is set to 2.
The precision parameter $\lambda$ is set to 2. 
$\DefEU_0$ and $q_i$'s are obtained by a binary search with a convex optimization as introduced in \cite{yang2012computing}.
The number of samples $\mathsf{T}$ used in \Sampling is set to 100.
In \Heuristic, \BF is used to calculate $\DefEU(U)$.

In the \texttt{QRI-MILP} algorithm, the optimal defender strategy is found with approximation parameter $K = 10$. The bi-level optimization algorithm is implemented using MATLAB R2017a. The low-level linear program is solved using the \texttt{linprog} function and the high-level optimization is solved with the \texttt{fmincon} function.
\section{Defending Against Level-$\kappa$ Attackers}
\label{app:levelk}
In the section \ref{Sec:level0}, we deal with the case with only type-0 attackers and provide algorithms to find the optimal set of informants to recruit. In this section, we show how those approaches can be easily extended to the case with level-$\kappa$ attackers.

Once given $\hat{\bx}^{\kappa-1}$, $\bq^{\kappa}$ can be easily obtained. So as $\DefEU_{\kappa}$, the defender's expected utility using $\bx_0$ against a single attack of a level-$\kappa$ attacker. To get the solution, we simply replace $\DefEU_{0},\bq^0$ with $\DefEU_{\kappa},\bq^{\kappa}$ and apply \Select or \Heuristic.
In order to calculate $\hat{\bx}^{\kappa-1}$ by definition, all that remains is to calculate $\MS(\mathbf{x}_0,\mathbf{x},\mathbf{q}^{i})$ for $i<\kappa$. The marginal probability of each target being covered can be calculated in a way similar to EDPA.

\end{document}